\def\ArxivPreprint{1}
\documentclass[11pt]{article}

\usepackage[margin=1in]{geometry}
\usepackage[T1]{fontenc}
\usepackage[utf8]{inputenc}
\usepackage{lmodern}
\usepackage{microtype}
\usepackage{setspace}
\ifdefined\UTMDWorkingPaper
\else
\fi

\usepackage{amsmath, amsthm, mathtools}

\usepackage{amssymb}
\allowdisplaybreaks

\usepackage{graphicx}
\usepackage{subcaption}
\usepackage{float} % for [H] placement
\usepackage[section]{placeins} % keep floats within their section
\usepackage{tikz}
\usetikzlibrary{arrows.meta,positioning,calc}
\usepackage{pgfplots}
\pgfplotsset{compat=1.18}
\usepgfplotslibrary{fillbetween}

\usepackage{booktabs}
\usepackage{array}
\usepackage[ruled]{algorithm2e}
\usepackage{enumitem}

\usepackage[authoryear,round]{natbib}
\ifdefined\UTMDWorkingPaper
\usepackage[colorlinks=true,allcolors=blue,hyperfootnotes=false]{hyperref}
\else
\usepackage[colorlinks=true,allcolors=blue]{hyperref}
\fi

\SetAlFnt{\small}
\SetAlCapFnt{\small}
\SetAlCapNameFnt{\small}
\SetAlCapHSkip{0pt}
\IncMargin{-\parindent}
\newcommand{\PaperTitle}{Exact Budget Balance via Payment-Rule Ambiguity: Incentive Preservation, Transfer Capacity, and Participation}
\ifdefined\UTMDWorkingPaper
\newcommand{\PaperDate}{July 28, 2026}
\else
\ifdefined\ArxivPreprint
\newcommand{\PaperDate}{August 1, 2026}
\else
\newcommand{\PaperDate}{}
\fi
\fi
\newcommand{\JELCodes}{D82, D81, D47, D44, C72}
\newcommand{\PaperKeywords}{Payment-rule ambiguity; Exact budget-balance transformations; Max--min expected utility; Mechanism design; Vickrey--Clarke--Groves mechanisms}
\hypersetup{
  pdftitle={\PaperTitle},
  pdfauthor={Hiroaki Odahara},
  pdfsubject={Mechanism design under payment-rule ambiguity},
  pdfkeywords={\PaperKeywords}
}

\title{\PaperTitle}
\makeatletter
\renewcommand{\maketitle}{%
  \vspace{-1.2em}%
  \begingroup
  \renewcommand\thefootnote{\fnsymbol{footnote}}%
  \begin{center}%
    {\LARGE\bfseries \@title\par}%
    \vskip 0.5em%
    {\@author\par}%
    \vskip 0.3em%
    {\@date\par}%
  \end{center}%
  \footnotesize\@thanks
  \endgroup
  \setcounter{footnote}{0}%
  \par\vskip 1em%
}
\makeatother
\ifdefined\UTMDWorkingPaper
\author{%
  \normalsize Hiroaki Odahara%
  \thanks{Graduate School of Economics, The University of Tokyo; Department of Informatics, Graduate School of Informatics and Engineering, The University of Electro-Communications.
  \texttt{hiroaki.odahara@e.u-tokyo.ac.jp}.
  This work was supported by JSPS KAKENHI Grant Number JP21H04979 and JST ERATO Grant Number JPMJER2301, Japan.
  I thank Hitoshi Matsushima and Atsushi Iwasaki for guidance and helpful discussions.
  This version substantially revises and expands UTMD Working Paper No.~130, ``Exact Budget-Balance Transformations via Payment-Rule Ambiguity'' (April 2026).
  Any remaining errors are my own.%
  }\\[0.2em]
}
\else
\ifdefined\ArxivPreprint
\author{%
  \normalsize Hiroaki Odahara\textsuperscript{1,2}\\[0.25em]
  \small \textsuperscript{1}Market Design Center, Graduate School of Economics, The University of Tokyo\\
  \small \textsuperscript{2}Graduate School of Informatics and Engineering, The University of Electro-Communications\\
  \small \texttt{hiroaki.odahara@e.u-tokyo.ac.jp}
}
\else
\author{%
  \normalsize Hiroaki Odahara\textsuperscript{a,b}\\[0.25em]
  \small \textsuperscript{a}Graduate School of Economics, The University of Tokyo,\\
  \small 7-3-1 Hongo, Bunkyo-ku, Tokyo 113-0033, Japan\\
  \small \textsuperscript{b}Department of Informatics, Graduate School of Informatics and Engineering,\\
  \small The University of Electro-Communications,\\
  \small 1-5-1 Chofugaoka, Chofu, Tokyo 182-8585, Japan\\
  \small Corresponding author: \texttt{hiroaki.odahara@e.u-tokyo.ac.jp}
}
\fi
\fi
\date{\small \PaperDate}

\newcommand{\N}{\{1,\dots,n\}}
\newcommand{\OmegaSet}{\Omega}
\newcommand{\E}{\mathbb{E}}

\theoremstyle{plain}
\newtheorem{theorem}{Theorem}[section]
\newtheorem{proposition}[theorem]{Proposition}
\newtheorem{lemma}[theorem]{Lemma}
\newtheorem{corollary}[theorem]{Corollary}

\theoremstyle{definition}
\newtheorem{example}[theorem]{Example}
\newtheorem{definition}[theorem]{Definition}

\newtheorem{remark}[theorem]{Remark}

\begin{document}

\maketitle
\thispagestyle{plain}

\begin{abstract}
This paper asks whether the payment side of an incentive-compatible mechanism can be repaired so that every realized settlement balances exactly while the allocation and all report-by-report comparisons under the maintained worst-case evaluation remain unchanged. Before reports, the designer commits to a menu of payment rules and a report-blind selection protocol. A pointwise preservation condition then carries over dominant-strategy incentives. Under a full-support common reference, a precisely known label law permits repair only if the original rule already balances, whereas frequency uncertainty about even one non-worst label makes every non-deficit rule repairable. Under the common residual-ambiguity benchmark, every preserving, exactly balanced finite menu at a positive-surplus state requires more transfer capacity as certified frequency information becomes tighter. In the canonical menu, fixing the transfer cap makes the best attainable balance deteriorate toward the original surplus as ambiguity vanishes. When only the possible labels are certified, two labels suffice. Suitably balanced partitions can also retain the original cap for nonnegative payments and preserve individual rationality under every realized rule. This last guarantee is impossible in positive-revenue Vickrey states when each agent's preserving distribution has full support. Frequency information therefore changes feasibility, liquidity, and realized participation in distinct ways.
\end{abstract}

\noindent\textbf{JEL Codes:} \JELCodes 

\ifdefined\UTMDWorkingPaper
\else
\ifdefined\ArxivPreprint
\else
\newpage
\fi
\fi

% ========================================================================
\section{Introduction}\label{sec:intro}
% ========================================================================

Many incentive-compatible mechanisms collect a surplus even when a neutral operator, internal platform, or clearinghouse is not allowed to retain revenue. This paper asks whether the allocation and its established incentives can be left untouched while the payment side is repaired so that every realized settlement clears exactly. The question is most relevant when law, charter, or accounting rules make zero retention a hard settlement constraint; if residual retention is acceptable, conventional redistribution is the more behaviorally robust benchmark. The instrument studied here is ambiguity over a publicly posted menu of payment rules.

A Vickrey auction gives the basic intuition. Instead of applying a single payment rule, the designer commits before bids to a menu of balanced settlement rules and selects one after bids through a report-blind protocol. Different rules can be most costly for different bidders. The menu can therefore be constructed so that each bidder's worst-case payment equals her original Vickrey payment even though every realized settlement redistributes the auction revenue completely. A deterministic rebate rule would generally change incentives because a bidder may manipulate the revenue being returned. Here no component rule need be incentive compatible on its own; under worst-case evaluation, the compound menu preserves every report-by-report payoff comparison.

The commitment is to formulas and procedure rather than discretion after reports. The menu and selection mapping are fixed in advance, the label is selected afterward, and no signal observed before reporting narrows the relevant ambiguity set. Agents may face the same ambiguity set while having different worst cases because different formulas are most costly in different payment rows. The natural applications are one-shot, episodic, or nonstationary settlements; stable repetition or enforceable label-contingent contracts may instead require recalibration or remove the intended exposure.

The analysis starts from a pointwise preservation proposition. Suppose each agent has a report-independent worst-case distribution over labels that reproduces her original payment at every report profile. The compound and original mechanisms then give that agent the same evaluated payoff report by report. Dominant-strategy incentives and participation under worst-case evaluation therefore carry over. Pointwise preservation is deliberately stronger than interim preservation: it leaves every deviation comparison unchanged and aligns the transformation with rule-by-rule balance and statewise transfer caps.

Frequency information affects three margins differently.

First, feasibility is discontinuous. Under a full-support common reference, a precisely known common label law permits repair only of rules that already balance. Within the fixed-adverse-label architecture, uncertainty about even one non-adverse label makes every non-deficit rule repairable. At the support-only benchmark, feasibility requires only two labels.

Second, liquidity moves in the opposite direction. Under common residual ambiguity, an identity for arbitrary finite menus makes the required transfer capacity grow as ambiguity narrows. The canonical frontier is sharp: under a fixed cap, the worst settlement imbalance returns toward the original surplus. At the support-only benchmark, suitably balanced partitions retain the original cap for nonnegative payments.

Third, realized participation has its own boundary. The support-only construction never charges an agent more than the original rule, so ex post individual rationality carries over rule by rule. If every preserving distribution has full support, that guarantee is impossible in positive-revenue Vickrey states. A small amount of ambiguity is therefore enough for feasibility, but the stronger participation package appears only at the support-only boundary.

The three claims have different scopes. The feasibility characterization is relative to publicly fixed adverse labels and pointwise preservation; allowing a broader label architecture enlarges the feasible set. The inverse-exposure capacity bound is more general and applies to arbitrary finite menus without fixed adverse labels or dominance, whereas the sharp closed-form frontiers use their stated single-winner benchmark domains. The support-only case is a static boundary benchmark under maintained worst-case behavior, not a claim of robustness to every preference specification. In the Vickrey benchmark, a deviation can tie under the maintained worst case while improving every non-worst settlement. The analysis also shows that exact repair survives under sufficiently pessimistic mixtures of worst- and best-case evaluation, although the endpoint construction retains the original cap and preserves participation rule by rule only under pure max--min evaluation. The interior regimes measure the price of imposing certified positive frequency floors and require recalibration when those floors change.

\paragraph{Related literature.}
The closest classical payment-redesign papers hold the allocation rule or interim transfers fixed while changing settlement. \citet{AspremontCremerGerardVaret2004} identify an information condition under which a Bayesian mechanism can be modified to balance ex post without changing its decision rule. \citet{BorgersNorman2009} establish an equivalence between ex ante and ex post balance under independent types and interim participation constraints, and \citet{Le2026} studies ex post balanced mechanisms that preserve the allocation and interim transfer rules of a mechanism that balances in expectation. The present paper instead preserves each agent's worst-case payment report by report, requires balance under every realized label, and studies transfer capacity, capped imbalance, and incidence.

\citet{BodohCreed2012} derives payoff equivalence under max--min expected utility and applies it to auctions and budget-balanced bilateral trade. \citet{Wolitzky2016} derives a general necessary condition for implementation with max--min agents and characterizes efficient bilateral trade when agents have limited information about one another's values. Here ambiguity concerns settlement-rule labels, while the allocation and its report-by-report incentive comparisons are held fixed. The pointwise preservation proposition serves as a foundation for different questions: whether a public menu can repair every non-deficit rule, how menu size and frequency information determine transfer capacity, and where rule-by-rule participation fails. In particular, balance is imposed label by label rather than only on the evaluated compound mechanism.

The closest ambiguous-transfer benchmark is \citet{Guo2019}. Guo uses ambiguous transfers to create the incentives needed for efficient implementation under the Beliefs Determine Preferences (BDP) condition of \citet{Neeman2004}. This paper instead begins with an incentive-compatible mechanism and uses ambiguity only to repair payments. The distinction matters for the standard independent-private-values Vickrey--Clarke--Groves (VCG) benchmark, which lies outside Guo's BDP sufficient condition. Both approaches require every candidate transfer rule to balance ex post; the additional object here is a public, report-blind, auditable mapping from operational inputs to rule labels.

The classical benchmark is the tension between incentive compatibility and balance in Vickrey--Clarke--Groves mechanisms \citep{Vickrey1961,Clarke1971,Groves1973,GreenLaffont1979,MyersonSatterthwaite1983}. Bailey--Cavallo and related redistribution mechanisms preserve the Groves argument by making each agent's rebate independent of that agent's own report while seeking no-deficit or approximate balance \citep{Bailey1997,Cavallo2006,GuoConitzer2009,Moulin2009}. The support-only repair takes a different route: it permits own-report-dependent rebates only in formulas that do not determine that agent's worst-case evaluation.

Bayesian and dynamic approaches alter the informational or temporal structure instead. They include the Arrow--d'Aspremont--Gérard-Varet (AGV) mechanism, correlated-information transfer constructions, and a range of later Bayesian and dynamic mechanisms \citep{Arrow1979,AGV1979,Matsushima1991,Aoyagi1998,Mezzetti2004,Matsushima2007,KosenokSeverinov2008,GalavottiMutoOyama2011,AtheySegal2013}.

One strand uses ambiguity in optimal auctions, surplus extraction, and the design of mechanisms or communication devices \citep{BoseOzdenorenPape2006,BoseDaripa2009,BoseRenou2014,DiTillioKosMessner2017,HwangKohBaik2026}. Another studies robust or max--min implementation and incentive compatibility when beliefs or informational environments are incompletely specified \citep{BergemannMorris2005,deCastroLiuYannelis2017,DeCastroYannelis2018,GuoYannelis2021,Song2018,TangZhang2021,GuoYannelis2022,Song2023}. Ambiguous contracting provides a related design perspective \citep{DuttingEtAl2024}. These literatures study how informational robustness or ambiguity alters implementability or optimal design. The main branch here instead repairs the payment side of a given mechanism. An appendix extension connects the construction to belief-contingent transfers for efficient balanced implementation \citep{JohnsonPrattZeckhauser1990} and to surplus-extraction bets \citep{CremerMcLean1985,CremerMcLean1988,Guo2024}.

Section~\ref{sec:setup} opens with a motivating two-rule Vickrey repair, then defines the general payment-rule mechanism and establishes pointwise incentive preservation. Section~\ref{sec:main} characterizes exact balance across frequency-information regimes. Section~\ref{sec:price} derives transfer-cap frontiers and fixed-cap bounds. Section~\ref{sec:transfers} establishes the participation boundary, and Section~\ref{sec:discussion} concludes. The appendices collect the full capped-repair frontier, heterogeneous-exposure incidence, quantitative participation bounds, and technical extensions and proofs.

% ========================================================================
\section{Payment-Side Repair and Incentive Preservation}\label{sec:setup}
% ========================================================================

Starting from a direct mechanism with fixed allocation and reporting incentives, this section presents a Vickrey repair and then formalizes the menu, timing, payoff equivalence, and ambiguity foundation.

\subsection{A motivating two-rule Vickrey repair}\label{subsec:two_rule_vickrey}

To make the transformation concrete before introducing the general notation, consider a numerical example. Payments are positive when made by a bidder to the auctioneer and negative when they are rebates. The two payment formulas and the report-blind label-selection protocol are fixed before bids, and one label is selected only after bids.

The support-only benchmark certifies which labels are possible but imposes no positive lower bound on how often a label must be used. Under max--min evaluation, an agent therefore evaluates the largest payment assigned to her across the menu. Call a label \emph{adverse} for an agent when it gives that agent this largest payment. An adverse-label partition groups agents according to the label assigned to be adverse for them, with the assignment fixed before reports.

Consider a three-bidder second-price auction at the bid profile $(9,6,1)$. Bidder~1 wins, and the original Vickrey rule charges the winner the second-highest bid and charges each loser zero, so its payment vector is $x=(6,0,0)$ and its total revenue is $s=6$. More generally, at any bid profile write the original Vickrey payments as $x=(x_1,x_2,x_3)$ and their sum as $s$. Fix two labels before reports, with label $A$ adverse for bidder~1 and label $B$ adverse for bidders~2 and~3. The following two payment rules are functions of the bids through $x$ and $s$:
\[
  z^A=\left(x_1,\ x_2-\frac{s}{2},\ x_3-\frac{s}{2}\right),
  \qquad
  z^B=\left(x_1-s,\ x_2,\ x_3\right).
\]
Each rule balances exactly. Every bidder's payment under her adverse label equals her original Vickrey payment, while the other rule never charges more when $s\ge0$. The worst-case evaluated payment is therefore unchanged.

Table~\ref{tab:two_rule_vickrey} displays the repair at the stated bid profile.

\begin{table}[ht]
\centering
\caption{A two-rule support-only repair of a three-bidder Vickrey state}
\label{tab:two_rule_vickrey}
\begin{tabular}{@{}lrrr@{}}
\toprule
 & Original Vickrey rule & Rule $A$ & Rule $B$ \\
\midrule
Bidder 1 & $6$ & $6$ & $0$ \\
Bidder 2 & $0$ & $-3$ & $0$ \\
Bidder 3 & $0$ & $-3$ & $0$ \\
\midrule
Account total & $6$ & $0$ & $0$ \\
\bottomrule
\end{tabular}
\end{table}

The partition is tied to bidder identities rather than the realized ranking, and neither component rule need be incentive compatible in isolation. Proposition~\ref{prop:pointwise_incentive} formalizes the resulting report-by-report payoff equivalence. Section~\ref{sec:main} extends the construction to fixed partitions and establishes two-label minimality within the maintained class; Example~\ref{ex:vcg_vickrey3} later compares the same state with menus calibrated to positive frequency floors.

\subsection{Payment-rule menus and timing}\label{subsec:menus_timing}

There are $n \geq 2$ agents indexed by $i \in \N$.
Agent $i$ has a finite type and report space $\Omega_i$.
Let $\OmegaSet:=\prod_{i\in\N}\Omega_i$ and $\Omega_{-i}:=\prod_{j\neq i}\Omega_j$.
For a finite set $S$, let $\Delta(S)$ denote its probability simplex.
True type profiles are denoted by $\theta=(\theta_1,\dots,\theta_n)\in\OmegaSet$. When Bayesian and interim statements are invoked, the commonly known prior $p\in\Delta(\OmegaSet)$ has full support; the pointwise dominant-strategy results do not use a type distribution.
In the statewise transfer algebra below, $\omega\in\OmegaSet$ denotes a generic report profile (or a state at truthful reporting). When true types and deviations appear together, $\theta$ denotes the true type profile and $r_i\in\Omega_i$ denotes agent $i$'s report. Finiteness is used to write interim expectations as finite sums; the algebraic transformation itself is pointwise in the report profile.
Let $A$ be the allocation set, $g:\OmegaSet\to A$ an allocation rule, $x:\OmegaSet\to\mathbb R^n$ a payment rule, and $v_i:A\times\OmegaSet\to\mathbb R$ agent $i$'s valuation function. If the true profile is $\theta$ and the report profile is $r$, agent $i$'s quasi-linear payoff is
\[
  u_i(\theta,r;\,g,x)=v_i(g(r),\theta)-x_i(r),
\]
where a positive $x_i(r)$ is a payment from agent~$i$ to the mechanism.
Because payments enter utility with a minus sign, minimizing utility over rule-label distributions is equivalent to maximizing expected payment. Throughout, a worst-case distribution refers to this same object in the two equivalent representations.
Unless noted otherwise, I do \emph{not} assume private values; the VCG benchmark in Section~\ref{sec:main} specializes to private values.

The pair $(g,x)$ is the \emph{pure} direct mechanism. Let $J$ be a finite label set. For each $j\in J$, let $x^j$ be the payment vector under label $j$; for each agent $i$, let the report-independent set $\Delta q_i\subseteq\Delta(J)$ describe her ambiguity about label frequencies. An \emph{ambiguous payment-rule mechanism} is a tuple
\[
  \mathcal{M} \;=\; \bigl(g,\ \{x^j\}_{j\in J},\ \{\Delta q_i\}_{i\in\N}\bigr).
\]
The canonical interior analysis assigns one adverse label to each agent and sets $J=\N$; the support-only analysis in Section~\ref{sec:main} also allows a smaller fixed partition menu. Subscripts index agents and superscripts index labels.

The timing is as follows. First, the designer commits to the allocation rule $g$, the family $\{x^j\}_{j \in J}$, and a report-blind label-selection protocol. Second, agents report before observing the residual operational input, the realized label, or any signal that changes the ambiguity set relevant at the reporting date. Third, the operational input is realized, the allocation is made, and one rule label is selected. Finally, the public transfer formula attached to that label is applied.

\paragraph{Auditable protocol.}\label{rem:auditable_protocol}
The transfer formulas and the mapping from operational inputs to rule labels are fixed before reports. Any otherwise unspecified input must be external, logged, realized only after reports, and tamper-resistant: neither the designer nor reporting agents, acting alone or jointly, can choose, alter, suppress, or selectively trigger it to steer the label, and the mapping cannot use report data. This is a maintained application condition, not a consequence of exact balance. The realized input, label, and formula are verifiable ex post; only the input's probability law remains unspecified. Thus ambiguity concerns ex ante frequencies, not discretion after reports. More generally, conditional on every signal available at the reporting date, the maintained ambiguity set must still admit the same report-independent worst-case distribution at every report profile. This conditional stability, rather than report-blindness alone, is the condition used in Proposition~\ref{prop:pointwise_incentive}; neither a known probability law nor statistical independence is required.

For example, a consortium clearinghouse can map a signed external event record and transaction identifier to a settlement label under formulas, eligible records, and an audit rule fixed in advance. An independent auditor may certify the eligible-label support alone or a public baseline component that guarantees stated frequency floors, while a residual label component depends on an external, nonstationary input. The same certificate and mapping are available to the designer and participants; an attaining menu is calibrated to those public bounds, not to a privately known exact law. Ex post verification of realized inputs does not by itself certify that their probability law is stationary. If a stable law is certified or reliably learned, the narrower ambiguity regime applies. Exact balance gives the clearinghouse a zero net account under every label and hence no direct monetary stake in steering selection, although reputational, side-contracting, and other nonpecuniary incentives remain.

\begin{remark}[Behavioral scope]\label{rem:behavioral_scope}
The guarantees require agents, at the reporting date, to retain the announced ambiguity sets and to be unable to neutralize label exposure completely through enforceable contingent contracts. Learning or certification that narrows these sets calls for recalibration, while complete label-contingent insurance lies outside the model. The endpoint implications of this requirement are stated after the support-only menu has been defined.
\end{remark}

\subsection{Pointwise payoff equivalence}\label{subsec:pointwise_equivalence}

A payment-rule menu $\{x^j\}_{j\in J}$ satisfies \emph{ex post budget balance} if
\begin{equation}\label{eq:bb}
  \sum_{i\in\N}x_i^j(\omega)=0
  \quad\text{for every }j\in J\text{ and }\omega\in\OmegaSet.
\end{equation}
Given report-independent weights $q_i^\ast\in\Delta(J)$, the menu \emph{pointwise preserves} the original payment rule if
\begin{equation}\label{eq:preserve}
  \sum_{j\in J}q_i^{j\ast}x_i^j(\omega)=x_i(\omega)
  \quad\text{for every }i\in\N\text{ and }\omega\in\OmegaSet.
\end{equation}
The weights are behaviorally relevant when $q_i^\ast$ maximizes agent~$i$'s expected payment over her ambiguity set at every report profile.

A pure mechanism is \emph{dominant-strategy incentive compatible} (DSIC) if truthful reporting is optimal for every profile of the other agents' reports. The compound menu is \emph{max--min DSIC} when the corresponding comparison uses the worst-case payoff over label distributions. Normalize outside options to zero. A pure mechanism is \emph{ex post individually rational} if truthful participation gives every agent a nonnegative payoff at every type profile. The compound menu is \emph{max--min ex post individually rational} (max--min ex post IR) when that comparison evaluates the menu before its label is realized. It satisfies \emph{rule-by-rule ex post individual rationality} (rule-by-rule ex post IR) only if participation is nonnegative under every realized label; this is the stronger requirement studied in Section~\ref{sec:transfers}.

\begin{proposition}[Pointwise payoff equivalence]\label{prop:pointwise_incentive}
Suppose utility is quasi-linear, the allocation rule $g$ is label-independent, and each $q_i^\ast$ is a report-independent worst-case distribution for agent~$i$. If the menu satisfies pointwise preservation \eqref{eq:preserve}, then, for every true type profile $\theta$ and every report profile $r$,
\begin{equation}\label{eq:pointwise_preserve_payoff}
  \min_{q_i\in\Delta q_i}
  \left[v_i(g(r),\theta)-\sum_{j\in J}q_i^j x_i^j(r)\right]
  =v_i(g(r),\theta)-x_i(r).
\end{equation}
Consequently, $(g,x)$ is DSIC if and only if $(g,\{x^j\})$ is max--min DSIC, and $(g,x)$ is ex post individually rational if and only if the compound menu is max--min ex post individually rational.
\end{proposition}

\begin{proof}
Because the valuation term is label-independent, minimizing utility is equivalent to maximizing the expected payment. The assumed report-independent maximizer is $q_i^\ast$, and \eqref{eq:preserve} makes its expected payment equal to $x_i(r)$. This proves \eqref{eq:pointwise_preserve_payoff}. Applying the identity to every truthful report and deviation gives the DSIC equivalence; applying it at truthful reports gives the evaluated ex post participation equivalence.
\end{proof}

This proposition takes the report-independent worst-case distribution as a reduced-form hypothesis. Lemma~\ref{lem:qstar} later verifies this hypothesis for the lower-bound ambiguity sets and fixed-adverse-label dominance introduced in the next subsection. The proposition invokes neither a type distribution nor statistical independence between types and labels. It does not say that a component payment rule is DSIC. Bayesian and interim extensions require additional separability assumptions and are stated in Appendix~\ref{app:bayesian_extension} after the ambiguity foundation has been specified.

\begin{remark}[Why pointwise preservation]\label{rem:pointwise_preservation}
Pointwise preservation is natural for the dominant-strategy and VCG applications because it keeps every report-by-report incentive comparison unchanged and aligns with rule-by-rule balance and transfer caps. If preservation were imposed only on interim expected transfers, gains and losses could offset across type or report profiles. That relaxation could enlarge the feasible set and lower the minimum transfer cap. The frontiers in Section~\ref{sec:price} are therefore sharp within the pointwise transformation class; they are not lower bounds for every interim-preserving payment redesign.
\end{remark}

\subsection{Fixed adverse labels and behavioral scope}\label{subsec:fixed_adverse_scope}

The primary behavioral objects are report-independent preserving weights $q_i^\ast\in\Delta(\N)$. In the max--min foundation below they are worst-case label distributions. The exact characterization uses a common reference vector $k\in\Delta(\N)$ satisfying
\begin{equation}\label{eq:referencewedge}
  q_i^{j\ast} \le k^j \le q_j^{j\ast}
  \quad \text{for all } i \neq j,\; j \in \N.
\end{equation}
Thus non-adverse agents place no more than the reference weight on a label, while its adverse agent places at least that weight. Since all vectors are probabilities,
\[
  R := \sum_{i \in \N} (q_i^{i\ast} - k^i)
  = \sum_{i \in \N} q_i^{i\ast} - 1 \ge 0.
\]
The condition $R>0$ is the nondegeneracy condition.

\begin{definition}[Reference-wedge condition]\label{def:refwedge}
Fix preserving weights $\{q_i^\ast\}_{i \in \N} \subseteq \Delta(\N)$.
Say that these weights \emph{satisfy the reference-wedge condition} if there exists $k \in \Delta(\N)$ satisfying \eqref{eq:referencewedge} and
\[
  R = \sum_{i \in \N} (q_i^{i\ast} - k^i) > 0.
\]
\end{definition}

The next subsection derives these weights and a wedge witness from a lower-bound ambiguity model.

To keep label $i$ adverse for agent~$i$ in the transformed family, I impose the pointwise dominance restriction
\begin{equation}\label{eq:monotone}
  x_i^j(\omega) \le x_i^i(\omega) \quad \text{for all } i \in \N,\; j \neq i,\; \omega \in \OmegaSet.
\end{equation}
Thus label~$i$ is pointwise the most expensive one for agent~$i$. The restriction gives the class a stable institutional meaning: each agent's adverse formula is publicly identified before reports rather than reassigned after observing them. Together with the lower-bound geometry below, it also supports a report-independent worst-case minimizer. Appendix~\ref{app:scope_counterexample} shows that dropping dominance enlarges the algebraically feasible class, so the restriction is substantive rather than without loss.

Given a pure payment rule $x$, a canonical payment-rule menu $\{x^j\}_{j\in\N}$ is a \emph{budget-balance (BB) transformation} of $x$ if it satisfies fixed-adverse-label dominance \eqref{eq:monotone}, ex post budget balance \eqref{eq:bb}, and pointwise preservation \eqref{eq:preserve}.

Some heterogeneity across preserving weights is essential. Under a common weight, only an already balanced pure rule can be preserved and balanced exactly.

\begin{proposition}[Necessity of heterogeneous preserving distributions]\label{prop:common_preserving_distribution}
Suppose all agents share the same report-independent preserving distribution $\bar q \in \Delta(\N)$ and that a payment-rule family $\{x^j\}_{j \in \N}$ satisfies \eqref{eq:bb} together with
\[
  \sum_{j \in \N} \bar q^j x_i^j(\omega) = x_i(\omega)
  \quad \text{for all } i \in \N,\; \omega \in \OmegaSet.
\]
Then the original pure payment rule is already exactly budget balanced state by state:
\[
  \sum_{i \in \N} x_i(\omega)=0
  \quad \text{for all } \omega \in \OmegaSet.
\]
\end{proposition}

\begin{proof}
Sum the preservation identities over agents and use \eqref{eq:bb}: $\sum_i x_i(\omega)=\sum_j\bar q^j\sum_i x_i^j(\omega)=0$ at every state.
\end{proof}

\subsection{A lower-bound foundation for the ambiguity sets}

I now derive the preserving weights from a lower-bound environment. The symmetric benchmark gives every agent the same common-contamination set, while a general lower-bound array identifies feasibility and the incidence of heterogeneous exposure. Later results use the resulting weights rather than the full ambiguity sets.

The only ambiguity studied here concerns the realized payment-rule label; uncertainty about other agents' types continues to be governed by the common prior $p$. Agent~$i$'s ambiguity about rule selection is represented by a nonempty convex set $\Delta q_i \subseteq \Delta(\N)$. Given a reference vector $k$ and lower-bound parameters $\varepsilon$, agent $i$ knows that each rule $j$ is selected with probability at least $(1 - \varepsilon_i^j) k^j$, where $\varepsilon_i^j \in [0,1]$, so that
\begin{equation}\label{eq:Deltaqi}
  \Delta q_i := \Bigl\{ q_i \in \Delta(\N) :\; q_i^j \ge (1-\varepsilon_i^j) k^j \text{ for all } j \in \N \Bigr\}.
\end{equation}
When $\varepsilon_i^j = 0$ for all $j$, agent~$i$ faces no ambiguity and knows the selection distribution exactly.
When $\varepsilon_i^j>0$ for at least one coordinate with $k^j>0$, the corresponding certified floor is strictly relaxed and the lower bounds no longer pin down the exact selection distribution.
The interpretation is incomplete disclosure of frequencies, not asymmetric information. The designer discloses the reference $k$ and commits to the lower bounds but does not specify an exact selection law. For a label with $k^j>0$, the parameter $\varepsilon_i^j$ measures how much frequency imprecision remains for agent $i$: a larger value weakens the certified floor, and the upper endpoint removes that floor altogether. A coordinate with $k^j=0$ imposes no floor at any parameter value. A protocol that instead announced an exact selection distribution would remove the ambiguity---and, by Proposition~\ref{prop:common_preserving_distribution}, the transformation with it. The public-certificate interpretation applies directly to the common lower bounds used in the main quantitative benchmark. Heterogeneous lower-bound arrays are a reduced-form extension and require participant-specific certified restrictions or admissible-model sets; they are not generated by a single common certificate alone.

The support-only benchmark is the opposite endpoint: no positive lower bound is imposed on a non-adverse label. Its sensitivity to narrower ambiguity sets is quantified in Appendix~\ref{app:calibration}.

As an interim representation, suppose the other agents report truthfully and let agent~$i$ evaluate the mechanism by max--min expected utility \citep{GilboaSchmeidler1989}:
\begin{equation}\label{eq:mmeu}
  \max_{r_i\in\Omega_i}\; \min_{q_i \in \Delta q_i}\;
  \E_p\Bigl[ v_i(g(r_i, \theta_{-i}), \theta)
  - \sum_{j \in \N} q_i^j\, x_i^j(r_i, \theta_{-i}) \;\Bigm|\; \theta_i \Bigr].
\end{equation}
The prior only averages these comparisons across the other agents' types. The dominant-strategy results use the pointwise worst-case payment comparison and do not depend on that prior. Each agent first considers the worst-case selection distribution within her ambiguity set and then chooses a report to maximize expected utility under that worst case.

Under the pointwise dominance restriction \eqref{eq:monotone}, label~$i$ is always the most expensive one for agent~$i$, so the inner minimization in \eqref{eq:mmeu} admits a simple report-independent minimizer:
\begin{equation}\label{eq:qstar}
  q_i^{j\ast} =
  \begin{cases}
    1 - \sum_{h \in \N \setminus\{i\}} (1 - \varepsilon_i^h) k^h, & j = i, \\
    (1 - \varepsilon_i^j) k^j, & j \neq i.
  \end{cases}
\end{equation}

\begin{lemma}[Report-independent worst-case distribution]\label{lem:qstar}
Under \eqref{eq:Deltaqi} and \eqref{eq:monotone}, for any report $r_i$ and any type $\theta_i$, the inner minimization in \eqref{eq:mmeu} admits $q_i^\ast$ in \eqref{eq:qstar} as a minimizer.
In particular, a worst-case distribution can be chosen report-independently.
\end{lemma}

\begin{proof}
The feasible set $\Delta q_i$ is a nonempty compact convex subset of the simplex $\Delta(\N)$, and for fixed $(r_i,\theta_{-i})$ the transfer term $\sum_j q_i^j x_i^j(r_i, \theta_{-i})$ is linear in $q_i$. Hence the inner minimization in \eqref{eq:mmeu} attains its optimum at an extreme point of $\Delta q_i$. By \eqref{eq:monotone}, label $i$ yields the largest transfer for agent~$i$ at every report profile, so one minimizing belief places as much probability mass as the constraints in \eqref{eq:Deltaqi} allow on label $i$ and the minimum feasible mass on every label $j \neq i$. This gives \eqref{eq:qstar}, and this minimizer does not depend on the report.
\end{proof}

Under \eqref{eq:qstar}, the off-diagonal components satisfy $q_i^{j\ast}=(1-\varepsilon_i^j)k^j \le k^j$ for every $j \neq i$, while
\[
  q_j^{j\ast}
  = 1 - \sum_{h \neq j}(1-\varepsilon_j^h)k^h
  = k^j + \sum_{h \neq j} \varepsilon_j^h k^h \ge k^j.
\]
Hence the lower-bound ambiguity model automatically implies the weak inequalities in \eqref{eq:referencewedge}, and
\[
  R = \sum_{j \in \N} (q_j^{j\ast} - k^j)
    = \sum_{j \in \N} k^j \sum_{i \neq j} \varepsilon_i^j.
\]
Thus the foundation satisfies the strict reference-wedge condition exactly when the displayed value of $R$ is positive. With full-support $k$, this is equivalent to $\varepsilon_i^j>0$ for some $i\neq j$; own-label exposure alone does not contribute to $R$.

The common-contamination benchmark gives every agent the familiar contamination set used in Knightian models such as \citet{EpsteinWang1994} and axiomatized by \citet{NishimuraOzaki2006},
\begin{equation}\label{eq:common_contamination_set}
  \Delta q_i(\varepsilon)
  :=\bigl\{(1-\varepsilon)k+\varepsilon q:q\in\Delta(\N)\bigr\},
  \qquad \varepsilon\in[0,1].
\end{equation}
Let $e_i$ denote the degenerate distribution that assigns probability one to label~$i$. Under fixed-adverse-label dominance, the preserving worst-case weights are
\begin{equation}\label{eq:common_contamination_weights}
  q_i^\ast(\varepsilon)=(1-\varepsilon)k+\varepsilon e_i.
\end{equation}
Thus the agents' worst-case distributions differ even though their primitive ambiguity sets are identical: each payment row has a different publicly assigned most-expensive label. The heterogeneity in the worst cases is generated by the menu, not by different announced laws or ambiguity attitudes. A larger common exposure expands the common ambiguity set and weakens every certified frequency floor. At zero exposure, all agents know the common label distribution exactly. At the upper endpoint, only the menu of possible labels is specified.

\paragraph{Ambiguity-neutral agents.}\label{rem:amb_neutral}
If all of agent $i$'s exposures vanish, then $\Delta q_i=\{k\}$ and $q_i^\ast=k$. More generally, if all her off-diagonal exposures vanish, then $q_i^\ast=k$, even though own-label exposure can leave $\Delta q_i$ non-singleton. With full-support $k$, one agent with a positive off-diagonal exposure is enough to make the wedge strict; all other agents may be ambiguity neutral.

\section{Exact Balance Across Exposure Regimes}\label{sec:main}
% ========================================================================

\subsection{Support-only partition repair}\label{subsec:partition_repair}

\paragraph{Notation convention.}
The symbol $x$ denotes the original pure payment rule, and superscripts index label-contingent payment rules. The symbol $z$ is used when an arbitrary-label or fixed-partition menu must be distinguished from the canonical one-label-per-agent menu, and a hat marks the canonical support-only endpoint family. The scalar $R$ denotes reference-wedge thickness; the calligraphic symbol $\mathcal R$ used later denotes an accounting rebate.

The support-only endpoint does not require one label per agent. Let $J$ be a finite label set and let $a:\N\to J$ be a surjective adverse-label map fixed before reports and used at every report profile. Formally, support-only ambiguity means that each agent's ambiguity set is the closed simplex $\Delta(J)$ over the certified labels. Hence every degenerate distribution $e_j$ is admissible and the worst case is attained. Write $z^j$ for the candidate payment rule under label $j$ and define the original surplus by $s(\omega):=\sum_{h\in\N}x_h(\omega)$. The original rule is \emph{statewise non-deficit} when $s(\omega)\ge0$ at every report profile. Write
\[
  G_j:=a^{-1}(j),\qquad g_j:=|G_j|.
\]
Under the assigned adverse-label dominance condition, the degenerate distribution on label $a(i)$ is a worst-case distribution for agent~$i$. Thus preservation requires $z_i^{a(i)}(\omega)=x_i(\omega)$, and fixed-adverse-label dominance requires $z_i^j(\omega)\le x_i(\omega)$ for every label $j$.

\begin{theorem}[Support-only partition repair]\label{thm:support-only-partition-repair}
Suppose $n\ge2$ and the adverse-label map $a:\N\to J$ is fixed ex ante and surjective.
For any partition with at least two nonempty groups, a support-only menu satisfying preservation, fixed-adverse-label dominance, and exact balance at every label exists if and only if the original payment rule is statewise non-deficit.
One such menu is
\begin{equation}\label{eq:partition_feasibility}
  z_i^j(\omega)=
  \begin{cases}
    x_i(\omega), & i\in G_j,\\[2pt]
    x_i(\omega)-\dfrac{s(\omega)}{n-g_j}, & i\notin G_j,
  \end{cases}
  \qquad j\in J.
\end{equation}
If the surplus is positive at some report profile, no one-label support-only menu can preserve and balance the original payments there. Hence two labels are necessary and sufficient for positive-surplus repair within this fixed-adverse-label class.
\end{theorem}

\begin{proof}
Fix a report profile and suppress it. If such a menu exists, dominance and support-only preservation give $z_i^j\le z_i^{a(i)}=x_i$ for every $i,j$. Exact balance therefore implies $0=\sum_i z_i^j\le\sum_i x_i=s$, so non-deficit is necessary. Conversely, suppose $s\ge0$. Formula \eqref{eq:partition_feasibility} preserves each adverse-label payment and satisfies dominance. For every label $j$,
\[
  \sum_{i\in\N}z_i^j
  =\sum_{i\in G_j}x_i+\sum_{i\notin G_j}x_i-s=0.
\]
Because the partition is fixed, applying the same formula at every report profile defines one menu of payment functions. With one label, preservation fixes every entry at $x_i$, so exact balance requires $s=0$. Any nontrivial bipartition supplies the two-label construction when $s>0$.
\end{proof}

The theorem permits signed original payments. Nonnegativity is needed only for the sharper cap results that follow.

\subsection{Menu size and endpoint transfer capacity}\label{subsec:partition_cap}

For any menu at a fixed report profile, define its \emph{absolute-transfer cap} as $\max_{i\in\N,\,j\in J}|z_i^j|$. Now fix a report profile with $x_i\ge0$ for every agent. For a fixed partition, let $M:=\max_i x_i$ and $S_j:=\sum_{i\in G_j}x_i$.

\begin{proposition}[Sharp support-only cap for a fixed partition]\label{prop:support-only-partition-cap}
Among support-only menus that use the fixed adverse-label partition $a$, preserve $x$, satisfy dominance, and balance every label, denote the minimum absolute-transfer cap by $m_{\mathrm{SO}}^\ast(x;a)$, where $\mathrm{SO}$ stands for support only. Then
\begin{equation}\label{eq:partition_cap}
  m_{\mathrm{SO}}^\ast(x;a)
  =\max\left\{M,\ \max_{j\in J}\frac{S_j}{n-g_j}\right\}.
\end{equation}
It is attained by
\begin{equation}\label{eq:partition_cap_menu}
  z_i^j=
  \begin{cases}
    x_i, & i\in G_j,\\[2pt]
    -\dfrac{S_j}{n-g_j}, & i\notin G_j.
  \end{cases}
\end{equation}
\end{proposition}

\begin{proof}
Preservation fixes every adverse entry at $x_i$, so any cap is at least $M$. In column $j$, exact balance forces the $n-g_j$ outside-group payments to sum to $-S_j$; hence at least one has absolute value at least $S_j/(n-g_j)$. Formula \eqref{eq:partition_cap_menu} attains both lower bounds simultaneously. Its outside-group entries are nonpositive, so nonnegative original payments ensure dominance.
\end{proof}

The feasibility construction \eqref{eq:partition_feasibility} and the cap-attaining construction \eqref{eq:partition_cap_menu} solve different objectives and need not coincide.

\begin{corollary}[Uniform cap-preserving menu size]\label{cor:partition_parity}
A fixed partition retains the original cap $M$ for every nonnegative payment vector if and only if every group has size at most half the number of agents. Consequently, the smallest number of labels with this uniform guarantee is two when the number of agents is even and three when it is odd.
\end{corollary}

\begin{proof}
If $g_j\le n/2$, then $S_j\le g_jM\le(n-g_j)M$, and Proposition~\ref{prop:support-only-partition-cap} gives the cap $M$. If $g_j>n/2$, set every payment in $G_j$ equal to $M>0$ and every other payment to zero. Then $S_j/(n-g_j)>M$, so the uniform guarantee fails. A two-block partition has both blocks of size at most $n/2$ exactly when they have equal size. For odd $n\ge3$, three nonempty blocks of size at most $n/2$ exist.
\end{proof}

This corollary is a uniform statement over all nonnegative payment vectors for a fixed ex ante partition. It does not say that every individual odd-agent state requires three labels; a two-label partition can retain the original cap at particular states.

\subsection{The canonical endpoint menu}\label{subsec:canonical_endpoint}

The canonical menu takes $J=\N$ and $a(i)=i$. At the support-only endpoint, it can use more labels than the smallest feasible repair. Its appeal is instead that it is permutation-equivariant, assigns each agent a publicly fixed adverse label, gives equal rebates across all non-adverse labels, and connects directly to the common-exposure interior analysis developed below.

\begin{corollary}[Canonical support-only endpoint]\label{cor:frequency_unrestricted}
Suppose $n\ge2$, the pure payment rule is statewise non-deficit, and each ambiguity set contains the degenerate distribution $e_i$ on agent~$i$'s adverse label. The canonical menu
\begin{equation}\label{eq:frequency_unrestricted_repair}
  \widehat x_i^j(\omega)
  :=x_i(\omega)-\mathbf1\{j\neq i\}\frac{s(\omega)}{n-1}
\end{equation}
satisfies dominance, pointwise preservation, and exact balance under every label. Every realized payment is weakly below the corresponding original payment, so rule-by-rule ex post IR is inherited from an ex post individually rational original mechanism. At any state with nonnegative original payments, the menu has the minimum possible absolute-transfer cap, equal to the largest original payment.
\end{corollary}

\begin{proof}
This is Theorem~\ref{thm:support-only-partition-repair} for the singleton partition. Because $e_i$ is admissible and label $i$ is adverse, it is a report-independent worst-case distribution and the evaluated payment is $\widehat x_i^i=x_i$. Formula \eqref{eq:frequency_unrestricted_repair} also gives $\widehat x_i^j\le x_i$ under every label. For nonnegative $x$, its entries lie between minus and plus $M:=\max_i x_i$; Proposition~\ref{prop:support-only-partition-cap} supplies the matching lower bound $M$.
\end{proof}

The last conclusion concerns the minimum cap, not the number of labels. The endpoint package is specific to pure max--min evaluation; the Vickrey example below explains the behavioral limitation, and Appendix~\ref{subsec:homogeneous} provides recalibrated alternatives.

\subsection{Positive exposure and feasibility}\label{subsec:positive_exposure}

I now carry the canonical menu into the common-reference interior. Appendix~\ref{subsec:fullclassification} provides the complete reduced-form classification beyond this branch.

\begin{theorem}[Exact budget-balance transformation and closed form]\label{thm:bbtrans}
Fix report-independent preserving weights $\{q_i^\ast\}_{i \in \N} \subseteq \Delta(\N)$.
Suppose these weights satisfy the reference-wedge condition in the sense of Definition~\ref{def:refwedge}.
Fix a pure payment rule $x$.
There exists a BB transformation of $x$ (i.e., a payment-rule family $\{x^j\}_{j \in \N}$ satisfying \eqref{eq:bb}, \eqref{eq:preserve}, and \eqref{eq:monotone}) if and only if $x$ is \emph{statewise non-deficit}:
\begin{equation}\label{eq:nodef}
  \sum_{i \in \N} x_i(\omega) \ge 0 \quad \text{for all } \omega \in \OmegaSet.
\end{equation}
Whenever \eqref{eq:nodef} holds, the transformation
\begin{equation}\label{eq:closed_form}
  x_i^j(\omega)
  =
  x_i(\omega)
  +
  \frac{\mathbf{1}\{j=i\} - q_i^{i\ast}}{R} \, s(\omega)
\end{equation}
implements the repair, where $\mathbf{1}\{j=i\}$ is the indicator of $j=i$.
\end{theorem}

\begin{proof}
Fix a report profile and suppress it. For necessity, sum the preservation identities and subtract the $k$-weighted sum of the balance identities. Since probability rows sum to one,
\[
  s
  =
  \sum_{i\in\N}\sum_{j\neq i}
  (k^j-q_i^{j\ast})(x_i^i-x_i^j).
\]
Every coefficient is nonnegative under the reference wedge, and every spread is nonnegative under dominance, so $s\ge0$.

For sufficiency, suppose $s\ge0$ and use \eqref{eq:closed_form}. Every diagonal--off-diagonal spread equals $s/R\ge0$, proving dominance. For each agent,
\[
  \sum_{j\in\N}q_i^{j\ast}x_i^j
  =
  x_i+\frac{s}{R}
  \sum_{j\in\N}q_i^{j\ast}
  \bigl(\mathbf{1}\{j=i\}-q_i^{i\ast}\bigr)
  =x_i,
\]
so preservation holds. Finally, for each label $j$,
\[
  \sum_{i\in\N}x_i^j
  =
  s+\frac{1-\sum_iq_i^{i\ast}}{R}s
  =0.
\]
The argument is pointwise and therefore applies at every report profile.
\end{proof}

\begin{remark}[Scope of the if-and-only-if result]
Theorem~\ref{thm:bbtrans} is an if-and-only-if result within the BB-transformation class, whose fixed-adverse-label dominance restriction gives labels a stable institutional meaning and supports a report-independent worst-case minimizer under the lower-bound foundation. Its necessity direction is not an impossibility theorem for unrestricted ambiguous mechanisms. Appendix~\ref{app:scope_counterexample} gives a two-agent deficit repair that becomes feasible when dominance is dropped. A complete classification of the problem without dominance is not claimed here.
\end{remark}

The lower-bound foundation turns this reduced-form condition into an extensive-margin characterization. Only off-diagonal exposure matters: own-label exposure alone does not move the preserving weight away from the common reference. The next theorem is the heterogeneous reduced-form extension of the common-contamination benchmark; retaining it in the main text isolates the minimal ambiguity needed for repair.

\begin{theorem}[At least one exposed agent is necessary and sufficient]\label{thm:oneagent}
Let the preserving weights be generated by \eqref{eq:Deltaqi}--\eqref{eq:qstar} from a full-support common reference $k$ and arbitrary off-diagonal exposures $\varepsilon_i^j\in[0,1]$ for every $j\neq i$; own-label exposures are unrestricted and play no role in the result. Fix a pure payment rule $x$.
\begin{enumerate}
\item[(a)] If $\varepsilon_i^j=0$ for every $i$ and every $j\neq i$, a BB transformation of $x$ exists if and only if $x$ is statewise \emph{balanced}: $\sum_i x_i(\omega)=0$ for all $\omega$.
\item[(b)] If $\varepsilon_i^j>0$ for at least one pair $i,j$ with $j\neq i$, a BB transformation of $x$ exists if and only if $x$ is statewise non-deficit: $\sum_i x_i(\omega)\ge0$ for all $\omega$.
\end{enumerate}
Consequently, for a statewise non-deficit rule that raises a positive surplus somewhere, the existence of at least one agent with positive exposure to at least one off-diagonal label is necessary and sufficient for exact repair within the maintained class.
\end{theorem}

\begin{proof}
(a) If all off-diagonal exposures vanish, \eqref{eq:qstar} gives $q_i^\ast=k$ for every $i$, so necessity is Proposition~\ref{prop:common_preserving_distribution}. Sufficiency is the constant family $x^j\equiv x$, which satisfies \eqref{eq:monotone} with equality, \eqref{eq:preserve} trivially, and \eqref{eq:bb} because $\sum_i x_i(\omega)=0$ at every state.

(b) The expression following \eqref{eq:qstar} gives
\[
  R=\sum_{j\in\N}k^j\sum_{i\neq j}\varepsilon_i^j>0
\]
whenever some off-diagonal $\varepsilon_i^j>0$ and $k$ has full support. Hence the preserving weights satisfy the reference-wedge condition of Definition~\ref{def:refwedge} with witness $k$, and Theorem~\ref{thm:bbtrans} applies. The argument remains valid when an exposure is at its upper endpoint.
\end{proof}

The theorem separates the extensive margin from the intensive margins studied below. A single arbitrarily small positive off-diagonal exposure changes which pure rules are repairable; its magnitude and allocation across agents determine the transfer burden and incidence.

At the support-only endpoint, \eqref{eq:closed_form} becomes the canonical menu \eqref{eq:frequency_unrestricted_repair}. Under the lower-bound foundation, $R=\sum_jk^j\sum_{i\neq j}\varepsilon_i^j$: the scalar measures aggregate off-diagonal exposure and the construction is unavailable when it vanishes. The reference-wedge inequalities are needed for the necessity direction; the closed-form construction itself uses only $R>0$.

\subsection{Exact calibration}\label{subsec:exact_calibration}

Return to the common-contamination benchmark in \eqref{eq:common_contamination_set}--\eqref{eq:common_contamination_weights} and fix its reference distribution $k\in\Delta(\N)$.
When $k$ has full support, every label receives positive preserving weight for every interior exposure level. The next result separates exact calibration from approximate robustness.

\begin{proposition}[No nontrivial exact calibration across exposure levels]\label{prop:no-uniform-calibration}
Fix a full-support reference $k$ and two distinct exposure levels $\varepsilon,\varepsilon'\in[0,1]$. Suppose the same canonical menu satisfies fixed-adverse-label dominance and exact balance and pointwise preserves the same original payment rule at both $q_i^\ast(\varepsilon)$ and $q_i^\ast(\varepsilon')$ for every agent. Then every agent's payment is constant across labels, and the original payment rule is already statewise balanced.
\end{proposition}

\begin{proof}
Fix a report profile and an agent. Subtracting the two preservation identities and using $\varepsilon\neq\varepsilon'$ gives
\[
  x_i^i=\sum_{j\in\N}k^j x_i^j.
\]
Dominance makes every term on the right weakly below $x_i^i$, and full support of $k$ therefore forces $x_i^j=x_i^i$ for every label $j$. Preservation then gives $x_i^j=x_i$ throughout the row. Exact balance of any column implies $\sum_i x_i=0$. Applying the argument at every report profile proves the claim.
\end{proof}

Thus a nontrivial exact repair is calibrated to one exposure level. This proposition does not preclude recalibrating the menu when the exposure changes, nor does it rule out approximate preservation. Appendix~\ref{app:calibration} gives the exact sensitivity formula, and Proposition~\ref{prop:eta_robust} gives the approximate feasibility boundary.

The behavioral limitation is not all-or-nothing. Appendix Corollary~\ref{cor:alpha_repair} studies a constant mixture of worst- and best-case evaluation. At a positive-surplus state, exact repair remains possible precisely when the worst-case weight exceeds $1/n$, although the required spread diverges as that threshold is approached. The no-additional-cap and rule-by-rule participation package at the support-only endpoint remains specific to pure max--min.

\long\def\SecondarySpreadAndRobustness{%
\begin{proposition}[Minimal adverse-label spread]\label{prop:closedform_optimal}
Suppose the preserving weights satisfy the reference-wedge condition with witness $k$ and let $R=\sum_i q_i^{i\ast}-1>0$.
Fix a state $\omega$ with surplus $s(\omega)\ge0$.
For any BB transformation $\{x^j\}$ of $x$, define the maximum adverse-label spread at $\omega$ by
\[
  G(\omega;\{x^j\})
  :=
  \max_{i\in\N,\; j\neq i}
  \bigl(x_i^i(\omega)-x_i^j(\omega)\bigr).
\]
Then every BB transformation satisfies
\[
  G(\omega;\{x^j\})\ge \frac{s(\omega)}{R}.
\]
The closed-form transformation in \eqref{eq:closed_form} attains this lower bound at every state. Hence it minimizes the maximum adverse-label spread among all BB transformations.
\end{proposition}

\begin{proof}
Fix a state and suppress it. Subtracting the $k$-weighted balance identities from the sum of the preservation identities gives
\[
  s=
  \sum_{i\in\N}\sum_{j\neq i}(k^j-q_i^{j\ast})
  \bigl(x_i^i-x_i^j\bigr).
\]
The wedge makes every coefficient nonnegative, and their sum is $R$. Hence $s\le RG$. Under \eqref{eq:closed_form}, every spread equals $s/R$, so the bound is attained.
\end{proof}

Proposition~\ref{prop:calibration_sensitivity} raises a complementary robustness question: exact preservation holds only at the designer's calibration point, but how does feasibility change when small preservation errors are allowed? A uniform preservation tolerance moves the feasibility boundary continuously, at Lipschitz rate $n$.

\begin{proposition}[Robustness of feasibility to approximate preservation]\label{prop:eta_robust}
Fix $\eta\ge0$ and suppose the preserving weights satisfy the reference-wedge condition. Then a payment-rule family satisfying \eqref{eq:monotone}, \eqref{eq:bb}, and the $\eta$-approximate preservation condition
\[
  \Bigl|\sum_{j\in\N}q_i^{j\ast}x_i^j(\omega)-x_i(\omega)\Bigr|\;\le\;\eta
  \qquad\text{for all }i,\omega
\]
exists if and only if $s(\omega)\ge-n\eta$ for all $\omega$. At $\eta=0$ this is Theorem~\ref{thm:bbtrans}.
\end{proposition}

\begin{proof}
See Appendix~\ref{app:eta_robust}.
\end{proof}

}
\long\def\BayesianExtensionResult{%
For this extension, fix the full-support prior $p$ introduced in Section~\ref{sec:setup} and maintain that ambiguity about labels is separable from uncertainty about types. A pure mechanism is \emph{Bayesian incentive compatible} (BIC) if truthful reporting maximizes each type's interim expected utility. The menu is \emph{max--min BIC} when the same comparison first takes the worst case over label distributions. With outside options normalized to zero, \emph{interim individual rationality} (IIR) and \emph{max--min IIR} use the corresponding truthful interim payoffs. Report-blindness by itself says only that the label-selection mapping does not use reports; the separability maintained here is the additional condition that permits the pointwise worst-case distribution to be integrated over types.

\begin{proposition}[Bayesian and interim extension]\label{prop:preserveIC}
Fix an allocation rule $g$ and a prior $p$.
Suppose $\{x^j\}_{j \in \N}$ is a BB transformation of $x$ and, for each agent $i$, the weight $q_i^\ast$ is a pointwise worst-case minimizer at every report profile $r\in\OmegaSet$---the form delivered by Lemma~\ref{lem:qstar} under \eqref{eq:monotone}.
Then:
\begin{enumerate}
  \item[(a)] $(g, x)$ is BIC if and only if $(g, \{x^j\})$ is max--min BIC.
  \item[(b)] $(g, x)$ satisfies IIR if and only if $(g, \{x^j\})$ satisfies max--min IIR.
\end{enumerate}
\end{proposition}

\begin{proof}
See Appendix~\ref{app:preserveIC}.
\end{proof}

Pointwise minimization permits integration over types, so preservation leaves BIC and IIR comparisons unchanged. This conclusion uses separability, not merely report-blind label selection.

}
\long\def\HomogeneousAndClassificationResults{%
\subsection{Homogeneous payment evaluations}\label{subsec:homogeneous}

The closed-form BB transformation in Theorem~\ref{thm:bbtrans} is not tied to linear worst-case averaging. To state the broader algebra, let $\rho_i:\mathbb R^n\to\mathbb R$ be agent $i$'s certainty-equivalent evaluation of a label-contingent \emph{payment} vector; a larger value is therefore worse. Let $\mathbf 1$ denote the all-ones vector in $\mathbb R^n$. Assume that $\rho_i$ is finite, componentwise nondecreasing, independent of types and reports, and satisfies
\begin{align}
  \rho_i(0)&=0,\label{eq:rho_zero}\\
  \rho_i(z+c\mathbf 1)&=\rho_i(z)+c
    &&\text{for every }z\in\mathbb R^n,\ c\in\mathbb R,\label{eq:rho_cash}\\
  \rho_i(\lambda z)&=\lambda\rho_i(z)
    &&\text{for every }z\in\mathbb R^n,\ \lambda\ge0.\label{eq:rho_homogeneous}
\end{align}
Monotonicity gives the payment interpretation; the construction below uses cash invariance and positive homogeneity. Identifying label distributions with vectors, continue to write $e_i$ for the $i$th standard basis vector and define
\[
  a_i:=\rho_i(e_i),
  \qquad
  A_\rho:=\sum_{i\in\N}a_i.
\]

To fix the aggregation order for Bayesian incentive and participation notions, define the interim utility functional
\begin{equation}\label{eq:rho_interim_utility}
\begin{aligned}
  U_i^\rho(r_i\mid\theta_i)
  &:=
  \E_p\!\left[
    v_i\bigl(g(r_i,\theta_{-i}),\theta\bigr)
    \,\middle|\,\theta_i
  \right] \\
  &\quad-
  \rho_i\!\left(
    \left(
      \E_p\!\left[
        x_i^j(r_i,\theta_{-i})
        \,\middle|\,\theta_i
      \right]
    \right)_{j\in\N}
  \right).
\end{aligned}
\end{equation}
Thus Bayesian uncertainty is integrated out label by label before $\rho_i$ is applied. BIC means that truthful reporting maximizes \eqref{eq:rho_interim_utility}, and IIR evaluates \eqref{eq:rho_interim_utility} at truthful reporting. For DSIC and ex post IR, use the pointwise counterpart
\begin{equation}\label{eq:rho_pointwise_utility}
  u_i^\rho(r_i,r_{-i};\theta)
  :=
  v_i\bigl(g(r_i,r_{-i}),\theta\bigr)
  -
  \rho_i\!\left(
    \bigl(x_i^j(r_i,r_{-i})\bigr)_{j\in\N}
  \right).
\end{equation}
The two aggregation orders need not agree for a general joint evaluation of Bayesian risk and ambiguity; the separability and order in \eqref{eq:rho_interim_utility} are maintained assumptions.

\begin{proposition}[Repair under homogeneous payment evaluations]\label{prop:rho_repair}
Suppose $A_\rho>1$ and the pure payment rule $x$ is statewise non-deficit. Then
\begin{equation}\label{eq:rho_repair}
  x_i^j(\omega)
  :=
  x_i(\omega)
  +\frac{\mathbf 1\{j=i\}-a_i}{A_\rho-1}\,s(\omega)
\end{equation}
satisfies dominance \eqref{eq:monotone}, exact balance \eqref{eq:bb}, and
\[
  \rho_i\bigl((x_i^j(\omega))_{j\in\N}\bigr)=x_i(\omega)
  \qquad\text{for every }i,\omega.
\]
\end{proposition}

\begin{proof}
Fix a state and put $\tau:=s/(A_\rho-1)\ge0$. Agent $i$'s payment row is
\[
  x_i\mathbf1+\tau(e_i-a_i\mathbf1).
\]
Hence every diagonal--off-diagonal spread equals $\tau$, proving dominance. Cash invariance and positive homogeneity give
\[
  \rho_i\bigl(x_i\mathbf1+\tau(e_i-a_i\mathbf1)\bigr)
  =x_i-\tau a_i+\tau\rho_i(e_i)=x_i.
\]
Finally, every column sum is $s+\tau(1-A_\rho)=0$.
\end{proof}

Proposition~\ref{prop:rho_repair} is a sufficient closed-form result, not a full feasibility classification for arbitrary nonlinear evaluations. It nevertheless preserves the relevant behavioral comparisons under a transparent separability condition. The identity in the proposition makes \eqref{eq:rho_pointwise_utility} equal to the pure-mechanism payoff at every report profile, so DSIC and ex post participation are preserved. For \eqref{eq:rho_interim_utility}, conditional expectation leaves the payment row in the same affine form:
\[
  \left(
    \E_p[x_i^j(r_i,\theta_{-i})\mid\theta_i]
  \right)_{j\in\N}
  =
  \bar x_i(r_i,\theta_i)\mathbf1
  +\bar\tau(r_i,\theta_i)(e_i-a_i\mathbf1),
\]
where $\bar x_i(r_i,\theta_i):=\E_p[x_i(r_i,\theta_{-i})\mid\theta_i]$ and $\bar\tau(r_i,\theta_i):=\E_p[s(r_i,\theta_{-i})\mid\theta_i]/(A_\rho-1)\ge0$. Cash invariance and positive homogeneity therefore make the evaluated interim payment equal to $\bar x_i(r_i,\theta_i)$, preserving BIC and IIR. Correlation between labels and types, report-dependent functionals, or a nonseparable joint evaluation of the two uncertainties lies outside this result.

If every agent instead uses the same linear payment evaluation with weight vector $\bar q$, then $a_i=\bar q^i$ and $A_\rho=1$, so the construction degenerates. This is consistent with Proposition~\ref{prop:common_preserving_distribution}: a common precise label law cannot repair a positive surplus while preserving evaluated payments.

\enlargethispage{2pt}

The constant-coefficient alpha-maxmin model gives a sharp specialization. Let $\alpha\in[0,1]$ denote the weight on worst-case utility.\footnote{Some alpha-maxmin conventions attach $\alpha$ to the best-case utility term. The present parameterization attaches it to the worst-case term; equivalently, it multiplies the maximum payment in \eqref{eq:alpha_payment}.} The corresponding payment evaluation is
\begin{equation}\label{eq:alpha_payment}
  \rho_i^\alpha(z)
  :=
  \alpha\max_{q\in\Delta q_i}q\cdot z
  +(1-\alpha)\min_{q\in\Delta q_i}q\cdot z .
\end{equation}
This reverses the max and min relative to utility because payments enter utility negatively. The specification is the constant-weight alpha-maxmin case associated with \citet{GhirardatoMaccheroniMarinacci2004}.

\begin{corollary}[Alpha-maxmin threshold and optimal spread]\label{cor:alpha_repair}
Fix a reference distribution $k\in\Delta(\N)$. Suppose all agents have the common contamination set
\[
  \Delta q_i=\bigl\{(1-\varepsilon)k+\varepsilon q:q\in\Delta(\N)\bigr\},
  \qquad \varepsilon\in(0,1),
\]
and evaluate payments by \eqref{eq:alpha_payment}. At any state with surplus $s>0$, a dominant, exactly balanced family preserving every $\rho_i^\alpha$ exists if and only if
\[
  \alpha>\frac1n.
\]
When this condition holds, one such family is
\begin{equation}\label{eq:alpha_repair}
  x_i^j
  =
  x_i+
  \frac{\mathbf1\{j=i\}-(1-\varepsilon)k^i-\alpha\varepsilon}
       {\varepsilon(n\alpha-1)}\,s .
\end{equation}
Moreover, every feasible family satisfies
\[
  \max_{i\in\N}\left(x_i^i-\min_{j\in\N}x_i^j\right)
  \;\ge\;\frac{s}{\varepsilon(n\alpha-1)},
\]
and \eqref{eq:alpha_repair} attains equality.
\end{corollary}

\begin{proof}
See Appendix~\ref{app:alpha_repair}.
\end{proof}

At $\alpha=1$, \eqref{eq:alpha_repair} and the optimal spread reduce to the max--min closed form with $R=(n-1)\varepsilon$. As $\alpha$ approaches $1/n$ from above, the required spread diverges.

The stronger endpoint package in Corollary~\ref{cor:frequency_unrestricted} is specific to this pure max--min case. Under the full-simplex endpoint, evaluating its canonical payment row by alpha-maxmin gives
\[
  \rho_i^\alpha\bigl((\widehat x_i^j)_j\bigr)
  =
  x_i-\frac{1-\alpha}{n-1}s,
\]
which equals the original payment only when $\alpha=1$ or the original rule is already balanced. For $\alpha<1$, the alpha-maxmin preserving family is different and the no-additional-cap and rule-by-rule participation conclusions do not follow.

\begin{remark}[Reference-frequency incidence under alpha-maxmin]\label{rem:alpha_reference}
For \eqref{eq:alpha_repair},
\[
  x_i-\sum_jk^jx_i^j
  =\frac{(\alpha-k^i)s}{n\alpha-1}.
\]
These amounts sum to $s$. Under a uniform reference every agent has the same reference-frequency accounting rebate $s/n$. Under a nonuniform reference, an agent has a weak accounting rebate only when $k^i\le\alpha$; otherwise she faces a reference-frequency accounting surcharge. Thus the agent-by-agent identity in Proposition~\ref{prop:reference_compensation} should not be carried over unchanged to this broader preference class. Generic smooth-ambiguity evaluations such as \citet{KlibanoffMarinacciMukerji2005} and act-dependent ambiguity-attitude weights need not satisfy positive homogeneity and remain outside Proposition~\ref{prop:rho_repair}.
\end{remark}

\subsection{Complete reduced-form classification}\label{subsec:fullclassification}

The next result classifies feasibility for arbitrary report-independent preserving weights within the fixed-adverse-label, dominance-restricted class. It separates the common-reference branch used in the main text from the crossed betting branch in Subsection~\ref{sec:crossed}.

For a preserving-weight profile $q^\ast$, define two sets of reference vectors:
\[
K^+(q^\ast)
:=
\Bigl\{
k\in\Delta(\N):
q_i^{j\ast}\le k^j\le q_j^{j\ast}
\text{ for all } i\neq j,\; j\in\N
\Bigr\},
\]
and
\[
K^-(q^\ast)
:=
\Bigl\{
k\in\Delta(\N):
q_i^{j\ast}\ge k^j\ge q_j^{j\ast}
\text{ for all } i\neq j,\; j\in\N
\Bigr\}.
\]
The set $K^+$ corresponds to the reference-wedge ordering used above; $K^-$ is the reverse ordering. Both sets admit an elementary characterization in terms of the columnwise extremes of the preserving weights, which makes each branch of the classification below directly checkable. For each label $j$, write
\[
  M^j:=\max_{i\neq j}q_i^{j\ast},
  \qquad
  m^j:=\min_{i\neq j}q_i^{j\ast}.
\]

\begin{lemma}[Checkable characterization of the reference sets]\label{lem:wedgechar}
$K^+(q^\ast)\neq\varnothing$ if and only if $\sum_{j\in\N}M^j\le1$, and $K^-(q^\ast)\neq\varnothing$ if and only if $\sum_{j\in\N}m^j\ge1$.
\end{lemma}

\begin{proof}
See Appendix~\ref{app:wedge_details}.
\end{proof}

The complete linear-algebraic classification is as follows.

\begin{theorem}[Reduced-form feasibility]\label{thm:general_feasibility}
Fix a pure payment rule $x$ and preserving weights $\{q_i^\ast\}_{i\in\N}\subseteq\Delta(\N)$ that do not depend on reports.
For each state $\omega$, let $s(\omega):=\sum_i x_i(\omega)$.
A BB transformation of $x$ exists if and only if the following statewise condition holds for every $\omega$:
\begin{enumerate}
  \item[(i)] if $K^+(q^\ast)\neq\varnothing$ and $K^-(q^\ast)=\varnothing$, then $s(\omega)\ge0$;
  \item[(ii)] if $K^+(q^\ast)=\varnothing$ and $K^-(q^\ast)\neq\varnothing$, then $s(\omega)\le0$;
  \item[(iii)] if $K^+(q^\ast)\neq\varnothing$ and $K^-(q^\ast)\neq\varnothing$, then $s(\omega)=0$;
  \item[(iv)] if $K^+(q^\ast)=\varnothing$ and $K^-(q^\ast)=\varnothing$, then there is no aggregate-sign restriction on $s(\omega)$.
\end{enumerate}
In case (iv), the linear system defining a BB transformation is feasible at that state for every vector of original transfers $x(\omega)$.
\end{theorem}

\begin{proof}
See Appendix~\ref{app:general_feasibility}.
\end{proof}

Cases (i)--(iii) respectively recover the main branch, its sign-reversed mirror, and already-balanced rules. Case (iv) has no surplus-sign restriction but requires crossed weights and at least three agents; Subsection~\ref{sec:crossed} explains its betting interpretation.

}
These results distinguish three comparisons: menu-size minimality at the support-only endpoint, cap optimality within a fixed architecture, and the architecture-robust capacity bounds developed in Section~\ref{sec:price}.

\subsection{The VCG benchmark}\label{subsec:vcg}

The VCG mechanism is the canonical benchmark for the transformation; it follows \citet{Vickrey1961}, \citet{Clarke1971}, and \citet{Groves1973}.

Assume \emph{private values}: for every agent $i$ there is a function $\hat v_i:A\times\Omega_i\to\mathbb R$ such that $v_i(a,\omega)=\hat v_i(a,\omega_i)$ for all $(a,\omega)$. To keep notation light, write this value as $v_i(a,\omega_i)$ below.
Recall that $A$ denotes the set of feasible allocations. Assume that all welfare maxima used below, both with all agents and with any one agent omitted, are finite and attained. In particular, for every $\omega \in \OmegaSet$ the efficient set $\arg\max_{a \in A} \sum_{i \in \N} v_i(a, \omega_i)$ is nonempty, and for every $i$ the welfare maximum obtained after omitting agent $i$ is attained.
Fix an efficient selection $g^*(\omega) \in \arg\max_{a \in A} \sum_{i \in \N} v_i(a, \omega_i)$.
A \emph{VCG mechanism} $(g^*, x^{\mathrm{VCG}})$ consists of this efficient allocation rule and the payment rule
\begin{equation}\label{eq:vcg}
  x_i^{\mathrm{VCG}}(\omega) = \max_{a \in A} \sum_{j \neq i} v_j(a, \omega_j)
  - \sum_{j \neq i} v_j(g^*(\omega), \omega_j).
\end{equation}
Each agent's payment equals the externality she imposes on others.
It is well known that VCG is DSIC: truthful reporting is a dominant strategy regardless of others' behavior.

I impose the following nonnegative-valuation assumption for the VCG application:
\begin{equation}\label{eq:normalization}
  \text{For all } a \in A,\; i \in \N,\; \omega_i \in \Omega_i:\quad
  v_i(a, \omega_i) \ge 0.
\end{equation}
By \eqref{eq:vcg}, VCG payments are nonnegative---$x_i^{\mathrm{VCG}}(\omega) \ge 0$ for all $i, \omega$---because the maximand can be evaluated at $a=g^*(\omega)$. Moreover,
\[
  v_i(g^*(\omega), \omega_i) - x_i^{\mathrm{VCG}}(\omega)
  = \sum_{j \in \N} v_j(g^*(\omega), \omega_j)
    - \max_{a \in A} \sum_{j \neq i} v_j(a, \omega_j)
  \ge 0,
\]
because $g^*(\omega)$ maximizes total value and \eqref{eq:normalization} gives $v_i(a,\omega_i)\ge 0$ for every $a$. Hence the underlying pure VCG mechanism satisfies ex post individual rationality.
Thus the pure VCG payment rule is statewise non-deficit and the underlying mechanism is ex post individually rational.

\begin{corollary}[Max--min DSIC, exact BB, and max--min ex post IR for VCG]\label{cor:vcg_bb}
Consider the private-value VCG environment above: an efficient selection $g^*$ exists, payments are given by \eqref{eq:vcg}, and the normalization \eqref{eq:normalization} holds. Suppose the payment-rule ambiguity sets are generated by the lower-bound foundation \eqref{eq:Deltaqi}, with preserving weights given by \eqref{eq:qstar}. By Lemma~\ref{lem:qstar}, each $q_i^\ast$ is then a report-independent worst-case minimizer under \eqref{eq:monotone}. Suppose also that these weights satisfy the reference-wedge condition of Definition~\ref{def:refwedge}.
Then a BB transformation of the VCG payment rule $x^{\mathrm{VCG}}$ exists, and the resulting payment-rule selection mechanism $(g^*, \{x^j\})$ satisfies simultaneously:
\begin{enumerate}
  \item[(i)] max--min DSIC;
  \item[(ii)] ex post budget balance;
  \item[(iii)] max--min ex post IR.
\end{enumerate}
\end{corollary}

\begin{proof}
VCG payments are nonnegative, so Theorem~\ref{thm:bbtrans} supplies a BB transformation. Pure VCG is DSIC and, by \eqref{eq:normalization}, ex post IR. Lemma~\ref{lem:qstar} and Proposition~\ref{prop:pointwise_incentive} preserve both properties under max--min evaluation, giving max--min DSIC and max--min ex post IR.
\end{proof}

Corollary~\ref{cor:vcg_bb} is a possibility result, and its economic reach depends on frequency information. At the support-only endpoint, Corollary~\ref{cor:frequency_unrestricted} adds rule-by-rule ex post IR without expanding the absolute-transfer cap. In the full-support interior, a cap-minimizing family must obey the transfer frontier characterized in Theorem~\ref{thm:frontier}, while rule-by-rule ex post IR fails in standard Vickrey states by Theorem~\ref{thm:full_support_participation_impossibility}. The following example makes the contrast concrete.

\begin{example}[Three-bidder Vickrey auction]\label{ex:vcg_vickrey3}
Consider the canonical one-label-per-bidder menu in a single-item auction with three bidders and value profile $(v_1,v_2,v_3)=(9,6,1)$.
The Vickrey allocation gives the item to bidder~1, and the VCG payment vector is
\[
  x^{\mathrm{VCG}}=(6,0,0),
\]
so the mechanism collects a surplus of $s=6$.
Take a symmetric reference vector $k^1=k^2=k^3=\frac13$, and for each bidder $i$ let
\[
  \varepsilon_i^j = \varepsilon \in (0,1]
  \qquad \text{for every } i,j.
\]
Then \eqref{eq:qstar} yields the worst-case distribution
\[
  q_i^{i\ast}=\frac{1+2\varepsilon}{3},
  \qquad
  q_i^{j\ast}=\frac{1-\varepsilon}{3}
  \quad (j\neq i),
\]
and hence
\[
  R=\sum_{h\in\N} q_h^{h\ast}-1 = 2\varepsilon.
\]
Applying the closed-form transformation \eqref{eq:closed_form} to $x^{\mathrm{VCG}}$ gives
\[
  x_i^j
  =
  x_i^{\mathrm{VCG}}
  + \frac{\mathbf{1}\{j=i\} - (1+2\varepsilon)/3}{2\varepsilon}\cdot 6.
\]
The resulting exactly budget-balanced rules are reported in Table~\ref{tab:vickrey_repairs}.
\begin{table}[H]
\centering
\caption{Exactly balanced canonical three-label payments under support-only ambiguity and two interior exposure levels.}
\label{tab:vickrey_repairs}
\begin{tabular}{cccc}
\toprule
Rule label $j$ & Support-only & $\varepsilon=\frac12$ & $\varepsilon=\frac15$ \\
\midrule
1 & $(6,-3,-3)$ & $(8,-4,-4)$ & $(14,-7,-7)$ \\
2 & $(3,0,-3)$ & $(2,2,-4)$ & $(-1,8,-7)$ \\
3 & $(3,-3,0)$ & $(2,-4,2)$ & $(-1,-7,8)$ \\
\bottomrule
\end{tabular}
\end{table}
Each displayed payment vector sums to zero, so every realized rule is ex post budget balanced.
The worst-case evaluated transfers are $(6,0,0)$ in all three columns, so truthfulness and max--min participation are preserved. In the support-only column, no bidder ever pays more than under the original Vickrey rule, the cap is $6$, and rule-by-rule ex post participation holds. Imposing positive frequency floors changes the picture. At exposure $1/5$, rule~1 charges the winner $14$ for an object worth $9$, even though her preserved evaluated payment remains $6$. The two interior columns use the spread-minimizing closed form and are illustrative rather than cap-attaining; Theorem~\ref{thm:frontier} gives the exact interior cap frontier. The example therefore displays both the clean endpoint and the liquidity and participation costs that arise away from it.
\end{example}

\paragraph{Expected utility and non-worst-label tie-breaking.}
The incentive guarantee concerns the compound family under the maintained max--min criterion, not each component formula. In the canonical support-only menu, fix the other bids and let a losing bidder $i$ raise her bid so that the second price rises from $s$ to $s'>s$ while she remains a loser. Her payment stays zero under adverse label $i$ and falls from $-s/(n-1)$ to $-s'/(n-1)$ under every non-adverse label. A report-independent subjective distribution $\pi_i\in\Delta(\N)$ with $\pi_i^i<1$ therefore makes the deviation strictly profitable: its expected-utility gain is
\[
  \frac{(1-\pi_i^i)(s'-s)}{n-1}>0.
\]
No revelation of the label before bidding is needed. Under pure max--min evaluation, by contrast, label $i$ remains worst and the two reports have the same evaluated payment. The deviation therefore ties with truth under the maintained criterion while being strictly better under every non-worst label. Pure Vickrey also makes this loser-preserving report payoff-equivalent. Proposition~\ref{prop:pointwise_incentive} and max--min DSIC remain intact; what is weaker is robustness to expected utility or to refinements that give non-worst labels positive relevance.

Known departures from pure max--min can sometimes be accommodated by recalibrating the menu. Proposition~\ref{prop:rho_repair} covers homogeneous payment evaluations when their aggregate adverse-label coefficient exceeds one, and Corollary~\ref{cor:alpha_repair} gives the sharp alpha-maxmin threshold. If ambiguity attitudes are unknown, the endpoint family has no such uniform guarantee: exact preservation becomes an additional calibration requirement, as Proposition~\ref{prop:no-uniform-calibration} and Appendix~\ref{app:calibration} illustrate.

\paragraph{Contrast with Bailey--Cavallo redistribution.}
The preceding deviation is the flip side of the construction. Bailey--Cavallo redistribution follows the standard Groves route: agent $i$'s rebate is a function only of the other agents' reports, so it does not disturb agent $i$'s Groves incentive comparison \citep{Bailey1997,Cavallo2006}. The support-only repair instead permits an own-report-dependent rebate under a non-adverse label because pure max--min evaluation is determined by the adverse label.

In a single-item Vickrey auction, Cavallo's rebate to bidder $i$ is $v_{-i}^{(2)}/n$, where $v_{-i}^{(2)}$ is the second-highest bid among the other bidders. For the same bids, the rebates are $(1/3,1/3,2)$, so Cavallo's mechanism returns $8/3$ of the VCG revenue and leaves $10/3$ with the operator. It retains standard preferences and rule-by-rule participation. The ambiguous repair instead returns all revenue and balances every rule. At the support-only endpoint it also preserves rule-by-rule participation without increasing the cap; in the full-support interior it can require larger payments and violate rule-by-rule participation. The comparison isolates the trade-off between single-rule robustness and exact recycling under the maintained worst-case evaluation.

This is not a claim that eliminating operator revenue is always welfare improving. When retention is permitted, Bailey--Cavallo offers the more robust compromise. The ambiguous repair addresses institutions in which retaining the residual is itself infeasible, so the relevant alternative to exact recycling may be an unbalanced settlement or suspension of the original allocation rule.

Appendix~\ref{sec:crossed} gives the institutionally distinct crossed branch, which repairs deficits through disagreement about the selection device; none of the main economic guarantees relies on it.

% ========================================================================
\section{Transfer Capacity}\label{sec:price}
% ========================================================================

This section studies how positive frequency floors affect transfer capacity. It derives a lower bound for arbitrary finite menus, the exact frontier in the canonical architecture, and the fixed-cap limit as exposure vanishes. The appendices give the full capped frontier and heterogeneous incidence.

\subsection{A menu-size-robust lower bound}\label{subsec:arbitrary_menu_bound}

Let $J$ be any finite nonempty label set and let $k\in\Delta(J)$ be common across agents. Extend the common-contamination definition in \eqref{eq:common_contamination_set} to $J$. At exposure $\varepsilon$, agent~$i$'s worst-case evaluated payment from a row $z_i=(z_i^\ell)_{\ell\in J}$ is
\begin{equation}\label{eq:arbitrary_menu_evaluation}
  (1-\varepsilon)\sum_{\ell\in J}k^\ell z_i^\ell
  +\varepsilon\max_{\ell\in J}z_i^\ell.
\end{equation}
This representation does not assign adverse labels in advance.

\begin{proposition}[Inverse-exposure identity for arbitrary finite menus]\label{prop:arbitrary_menu_identity}
Fix a state, an original payment vector $x$, its surplus $s:=\sum_{i\in\N}x_i$, an exposure $\varepsilon\in(0,1]$, any finite label set $J$, and a common reference distribution $k\in\Delta(J)$. Suppose a payment-rule menu $\{z^\ell\}_{\ell\in J}$ balances exactly at every label and preserves each original payment in the sense that
\[
  (1-\varepsilon)\sum_{\ell\in J}k^\ell z_i^\ell
  +\varepsilon\max_{\ell\in J}z_i^\ell=x_i
  \qquad\text{for every }i\in\N.
\]
Then
\begin{equation}\label{eq:arbitrary_menu_identity}
  \sum_{i\in\N}\max_{\ell\in J}z_i^\ell=\frac{s}{\varepsilon}.
\end{equation}
Consequently, if $m:=\max_{i\in\N,\,\ell\in J}|z_i^\ell|$ is the menu's absolute-transfer cap, then
\begin{equation}\label{eq:arbitrary_menu_cap_bound}
  m\ge\frac{s}{n\varepsilon}.
\end{equation}
\end{proposition}

\begin{proof}
Sum the preservation identities over agents. The reference-weighted term vanishes because every label balances, leaving
\[
  s=(1-\varepsilon)\sum_{\ell\in J}k^\ell\sum_i z_i^\ell
    +\varepsilon\sum_i\max_{\ell\in J}z_i^\ell
   =\varepsilon\sum_i\max_{\ell\in J}z_i^\ell.
\]
This proves \eqref{eq:arbitrary_menu_identity}. If every absolute payment is at most $m$, each row maximum is at most $m$, which gives \eqref{eq:arbitrary_menu_cap_bound}.
\end{proof}

The identity needs neither one label per agent, fixed adverse labels, dominance, nonnegative original payments, nor a full-support reference. For negative surplus the displayed cap bound is vacuous, and the proposition does not establish attainability or a complete frontier for arbitrary menus. Its role is to show that the inverse-exposure burden cannot be removed merely by adding labels.

The lower bounds restrict every family that preserves the weights agents actually use and do not require the designer to know those weights. Attaining menus, by contrast, must be calibrated to the stated reference and exposure; Proposition~\ref{prop:no-uniform-calibration} and Appendix~\ref{app:calibration} give the exact and approximate limits.

\subsection{Canonical surplus absorption}\label{sec:identity}

The engine of this section is the observation that, under a common off-diagonal exposure, preservation yields an exact surplus-absorption identity rather than only a transfer-cap inequality.
For any payment-rule family, write
\[
  T_j(\omega):=\sum_{i\in\N}x_i^j(\omega)
\]
for the aggregate payment, or column imbalance, under label $j$ at state $\omega$.

\begin{lemma}[Surplus-absorption identity under common exposure]\label{lem:trace}
Fix a state $\omega$ and suppress it, and let $\varepsilon_i^j=\varepsilon\in[0,1)$ for all $i$ and $j\neq i$; the own-label exposure $\varepsilon_i^i$ is irrelevant for what follows. Then any payment-rule family satisfying the preservation condition \eqref{eq:preserve} at $\omega$ obeys
\[
  s\;=\;\sum_{j\in\N}k^jT_j\;+\;\frac{\varepsilon}{1-\varepsilon}\sum_{i\in\N}\bigl(x_i^i-x_i\bigr).
\]
In particular, if the family is exactly budget balanced at $\omega$ and $\varepsilon>0$, then
\[
  \sum_{i\in\N}x_i^i\;=\;\frac{s}{\varepsilon}\,.
\]
\end{lemma}

\begin{proof}
By \eqref{eq:qstar}, preservation for agent $i$ reads $(1-\varepsilon)\sum_{j\neq i}k^jx_i^j+\bigl[k^i+\varepsilon(1-k^i)\bigr]x_i^i=x_i$. Adding and subtracting $(1-\varepsilon)k^ix_i^i$ turns the diagonal coefficient into $\varepsilon$:
\[
  (1-\varepsilon)\sum_{j\in\N}k^jx_i^j+\varepsilon\,x_i^i=x_i,
  \qquad\text{i.e.,}\qquad
  \sum_{j\in\N}k^jx_i^j=\frac{x_i-\varepsilon x_i^i}{1-\varepsilon}.
\]
Summing over $i$ and exchanging the order of summation gives $\sum_jk^jT_j=\sum_i(x_i-\varepsilon x_i^i)/(1-\varepsilon)$, and therefore
\[
  \sum_{j\in\N}k^jT_j+\frac{\varepsilon}{1-\varepsilon}\sum_{i\in\N}(x_i^i-x_i)
  =\sum_{i\in\N}\frac{(x_i-\varepsilon x_i^i)+\varepsilon(x_i^i-x_i)}{1-\varepsilon}
  =\sum_{i\in\N}x_i=s.
\]
Under exact balance $T_j\equiv0$, so $\sum_i(x_i^i-x_i)=s(1-\varepsilon)/\varepsilon$, i.e., $\sum_ix_i^i=s/\varepsilon$.
\end{proof}

The exact-balance form is independent of the reference distribution: adverse-label payments must absorb the surplus grossed up by inverse exposure. This is an identity, not merely a bound, and only preservation is used before the imbalance term is set to zero.

\begin{corollary}[Sharpened cap--imbalance frontier]\label{cor:capfrontier}
In the setting of Lemma~\ref{lem:trace}, suppose $|T_j|\le\bar\delta$ for all $j$ and $x_i^i\le\bar m$ for all $i$. Then
\[
  s\;\le\;(1-\varepsilon)\,\bar\delta\;+\;n\varepsilon\,\bar m .
\]
\end{corollary}

\begin{proof}
The identity gives $s\le\bar\delta+\frac{\varepsilon}{1-\varepsilon}\sum_i(\bar m-x_i)=\bar\delta+\frac{\varepsilon}{1-\varepsilon}(n\bar m-s)$, since $\sum_jk^j|T_j|\le\bar\delta$ and $\sum_ix_i=s$. Rearranging yields the display.
\end{proof}

The identity need not be the only binding restriction. With three agents, the off-diagonal entries can impose an additional allocation bound. The following lemma isolates it for later use.

\begin{lemma}[Three-agent allocation bound]\label{lem:three_allocation}
Let $n=3$, $k^j\equiv1/3$, and $\varepsilon_i^j=\varepsilon\in(0,1)$ for all $i$ and $j\neq i$. Fix a single-winner state $x=(t,0,0)$ with $t>0$. Any preserving family whose entries satisfy $|x_i^j|\le\bar m$ obeys
\begin{equation}\label{eq:three_allocation_bound}
  \max_{j\in\N}|T_j|
  \;\ge\;
  \frac{(1+2\varepsilon)t-\varepsilon(4-\varepsilon)\bar m}
       {1-\varepsilon^2}.
\end{equation}
Dominance is not needed for this lower bound.
\end{lemma}

\begin{proof}
See Appendix~\ref{app:three_bound}.
\end{proof}

Economically, surplus can be absorbed through either rule-level imbalance or adverse-label transfer capacity. Exact balance closes the first channel, so positive frequency floors can raise the required transfer capacity. Corollary~\ref{cor:capfrontier} identifies this accounting constraint; the full single-winner frontier also incorporates the winner-preservation floor and, with three agents, an allocation bound. Proposition~\ref{prop:general_transfer_lb} gives the corresponding architecture-robust comparison.

\subsection{The exact transfer-cap frontier}\label{sec:frontier}

For a state $\omega$, define the minimal transfer cap compatible with exact repair---the formal liquidity price of exactness---
\[
  m^\ast(\omega)\;:=\;\inf\Bigl\{\,\max_{i,j\in\N}\bigl|x_i^j(\omega)\bigr|\;:\;\{x^j(\omega)\}_{j\in\N}\ \text{satisfies \eqref{eq:monotone}, \eqref{eq:bb}, and \eqref{eq:preserve} at }\omega\Bigr\},
\]
with $m^\ast(\omega)=+\infty$ if no such family exists. The benchmark state is \emph{single-winner}: $x(\omega)=(t,0,\dots,0)$ with $t>0$ after relabeling agents, so $s(\omega)=t$. This is exactly the transfer profile of a single-item Vickrey state.

\begin{theorem}[Exact transfer-cap frontier in the single-winner domain]\label{thm:frontier}
Let the preserving weights be generated by \eqref{eq:Deltaqi}--\eqref{eq:qstar}. Let $n\ge3$ and $\varepsilon_i^j=\varepsilon\in(0,1)$ for all $i$ and $j\neq i$, and fix a single-winner state $\omega$.
\begin{enumerate}
\item[(i)] For every full-support reference $k$, every family satisfying \eqref{eq:bb} and \eqref{eq:preserve} at $\omega$ has
\[
  \max_{i,j\in\N}\bigl|x_i^j(\omega)\bigr|\;\ge\;\max\Bigl\{\,t,\ \frac{t}{n\varepsilon}\Bigr\}.
\]
\item[(ii)] If $k^j\equiv1/n$, the frontier over the entire interior exposure range is
\begin{equation}\label{eq:full_exact_frontier}
m^\ast(\omega)=
\begin{cases}
\displaystyle \max\Bigl\{t,\frac{t}{n\varepsilon}\Bigr\},
  & n\ge4,\\[7pt]
\displaystyle \max\Bigl\{\frac{t}{3\varepsilon},
  \frac{(1+2\varepsilon)t}{\varepsilon(4-\varepsilon)}\Bigr\},
  & n=3.
\end{cases}
\end{equation}
Every branch of \eqref{eq:full_exact_frontier} is attained. Appendix~\ref{app:frontier} provides explicit equal-split, thick-exposure, and three-agent attaining families.
\end{enumerate}
\end{theorem}

\begin{proof}
See Appendix~\ref{app:frontier}.
\end{proof}

The winner's preservation constraint requires some entry of her row to reach $t$, while Lemma~\ref{lem:trace} forces the diagonal to sum to $t/\varepsilon$. For four or more agents, equal splitting attains the diagonal bound until the winner floor $t$ takes over; the thick-exposure family then attains that floor. With three agents, Lemma~\ref{lem:three_allocation} adds an off-diagonal allocation bound. The second term in the three-agent formula already weakly exceeds the winner floor because $(1+2\varepsilon)/[\varepsilon(4-\varepsilon)]\ge1$ is equivalent to $(1-\varepsilon)^2\ge0$; this is why $t$ does not appear separately. At $\varepsilon=1/7$, this constraint replaces equal diagonal splitting as the binding one.

The interior frontier connects continuously to the support-only endpoint. In a single-winner state, Corollary~\ref{cor:frequency_unrestricted} gives $m^\ast(\omega)=t$ at that endpoint for every number of agents. More generally, its endpoint cap formula applies to every componentwise nonnegative payment vector, not only to single-winner states.

For $n\ge4$, the interior frontier already equals the original winner payment $t$ once $\varepsilon\ge1/n$. The full-support participation impossibility in Theorem~\ref{thm:full_support_participation_impossibility} nevertheless remains. In this benchmark, adding agents can therefore eliminate the additional liquidity cost without restoring rule-by-rule participation.

On the equal-split branch, the single-winner state is not special. Because the feasible set is defined by linear equalities and inequalities, convexity propagates that branch to \emph{every} nonnegative transfer profile. Outside this domain, mixing single-winner constructions still gives upper bounds, but the exact cap need not depend only on the aggregate surplus.

Identifying label distributions with vectors in $\mathbb R^n$, $e_i$ is the $i$th standard basis vector.

For the equal-split and surplus-only extensions, define
\[
  \hat\varepsilon(3):=\tfrac17,
  \qquad
  \hat\varepsilon(n):=\tfrac1n\quad(n\ge4).
\]

\begin{corollary}[Surplus-only exact frontier on the equal-split domain]\label{cor:anyx}
Let the preserving weights be generated by \eqref{eq:Deltaqi}--\eqref{eq:qstar}. Let $n\ge3$, $k^j\equiv1/n$, and $\varepsilon_i^j=\varepsilon\in(0,\hat\varepsilon(n)]$ for all $j\neq i$. Then for every state $\omega$ whose transfer vector is componentwise nonnegative, $x_i(\omega)\ge0$ for all $i$,
\[
  m^\ast(\omega)=\frac{s(\omega)}{n\varepsilon}\qquad\text{whenever }s(\omega)>0,
\]
and $m^\ast(\omega)=0$ when $s(\omega)=0$. The minimal cap is therefore the same for every distribution of a given surplus across agents; and since VCG payments are componentwise nonnegative, the frontier prices the exact repair of every VCG state.
\end{corollary}

\begin{proof}
Write $x:=x(\omega)$ and $s:=s(\omega)$. The lower bound is Lemma~\ref{lem:trace}: exact balance forces $\sum_ix_i^i=s/\varepsilon$, so $\max_ix_i^i\ge s/(n\varepsilon)$. For the upper bound, note that $m^\ast(\omega)$ depends on the state only through the vector $x(\omega)$, and that the map $x\mapsto m^\ast(x)$ is convex: if $\{z^j\}$ and $\{\tilde z^j\}$ satisfy \eqref{eq:monotone}, \eqref{eq:bb}, and \eqref{eq:preserve} for the vectors $x$ and $\tilde x$ with caps $m$ and $\tilde m$, then the entrywise combination $\lambda z+(1-\lambda)\tilde z$ satisfies the same linear equalities and inequalities for $\lambda x+(1-\lambda)\tilde x$, with cap at most $\lambda m+(1-\lambda)\tilde m$. Because both the uniform reference and the common exposure profile are invariant under simultaneous permutations of agents and labels, the single-winner construction of Theorem~\ref{thm:frontier}(ii) applies to the vector $s\,e_i$ for every agent $i$, so $m^\ast(s\,e_i)=s/(n\varepsilon)$. Writing a nonnegative $x$ with $\sum_ix_i=s>0$ as the convex combination $x=\sum_i(x_i/s)(s\,e_i)$ and applying convexity yields $m^\ast(x)\le\sum_i(x_i/s)\cdot s/(n\varepsilon)=s/(n\varepsilon)$. If $s=0$ and $x\ge0$, then $x=0$ and the zero family is feasible.
\end{proof}

Appendix~\ref{app:price_supp} studies approximate preservation, compares the spread and cap criteria, and gives the linear-programming formulation.

\subsection{Fixed-cap imbalance and vanishing exposure}\label{sec:capped}

Theorem~\ref{thm:frontier} prices exactness in the interior: it reports the least transfer cap under which exact balance survives. Institutions often face the converse problem. Limited liability, margin requirements, or accounting conventions fix a hard cap on realized transfers, and the designer's objective is to come as close to balance as the cap allows. This subsection defines that capped repair problem and records the general lower bound needed for the vanishing-exposure result. Appendix~\ref{app:full_capped_frontier} derives the full sharp frontier and its attaining families. The support-only endpoint is treated separately by Corollary~\ref{cor:frequency_unrestricted}.

For a state $\omega$ and a cap $\bar m\ge0$, define the \emph{capped repair value}
\begin{equation}\label{eq:deltastar}
\begin{aligned}
  \delta^\ast_{\bar m}(\omega):=
  \inf\Bigl\{\,&\max_{j\in\N}\bigl|T_j(\omega)\bigr|:\,
  \{x^j(\omega)\}_{j\in\N}\ \text{satisfies \eqref{eq:monotone} and \eqref{eq:preserve} at $\omega$,}\\
  &\max_{i,j\in\N}\bigl|x_i^j(\omega)\bigr|\le\bar m\Bigr\}.
\end{aligned}
\end{equation}
Set $\delta^\ast_{\bar m}(\omega):=+\infty$ when the constraint set is empty. This is the minimum worst-rule absolute imbalance among capped preserving families. It is the companion of the minimal cap: since the constraint set in \eqref{eq:deltastar} is compact, the infimum is attained whenever finite, so $m^\ast(\omega)=\inf\{\bar m\ge0:\delta^\ast_{\bar m}(\omega)=0\}$; Remark~\ref{rem:capped_lp} describes the common value surface.

Write $[z]_+:=\max\{z,0\}$.

\long\def\FullCappedFrontier{%
\subsection{Full capped-repair frontier}\label{app:full_capped_frontier}

A deliberately simple family gives a first upper bound: scale the closed form. For $\gamma\in[0,1]$, define
\begin{equation}\label{eq:scaled}
  x_i^j(\omega;\gamma):=x_i(\omega)+\gamma\,\frac{\mathbf{1}\{j=i\}-q_i^{i\ast}}{R}\,s(\omega).
\end{equation}

The scaled family preserves evaluations and dominance and has imbalance $T_j(\omega;\gamma)=(1-\gamma)s(\omega)$ under every label. It is a convenient but generally nonoptimal benchmark; Proposition~\ref{prop:scaled} gives its cap bound, and Remark~\ref{rem:capped_position} gives the gap.

Part (i) below gives a universal lower bound driven by the surplus-absorption identity. For four or more agents, part (ii) shows that this accounting bound is tight over the entire interior exposure range. With three agents, part (iii) combines it with Lemma~\ref{lem:three_allocation}; the additional allocation line creates a kink only on the thin side of the parameter space.

\begin{theorem}[Capped repair frontier in the single-winner domain]\label{thm:capped_frontier}
Let the preserving weights be generated by \eqref{eq:Deltaqi}--\eqref{eq:qstar}. Let $n\ge3$ and $\varepsilon_i^j=\varepsilon\in(0,1)$ for all $i$ and $j\neq i$, and fix a single-winner state $\omega$ with surplus $t>0$.
\begin{enumerate}
\item[(i)] For every full-support reference $k$: if $\bar m<t$, the constraint set in \eqref{eq:deltastar} is empty and $\delta^\ast_{\bar m}(\omega)=+\infty$; and every family satisfying \eqref{eq:preserve} at $\omega$ with $x_i^i(\omega)\le\bar m$ for all $i$ obeys
\[
  \max_{j\in\N}\bigl|T_j(\omega)\bigr|
  \;\ge\;\frac{\bigl[t-n\varepsilon\bar m\bigr]_+}{1-\varepsilon}\,.
\]
\item[(ii)] Let $n\ge4$ and $k^j\equiv1/n$. Then
\[
  \delta^\ast_{\bar m}(\omega)=
  \begin{cases}
    +\infty, & 0\le\bar m<t,\\[3pt]
    \dfrac{\bigl[t-n\varepsilon\bar m\bigr]_+}{1-\varepsilon},
      & \bar m\ge t.
\end{cases}
\]
Every finite branch is attained. Appendix~\ref{app:capped} gives the capped equal-split construction for the positive branch and the exact construction used once the displayed value reaches zero.
\item[(iii)] Let $n=3$ and $k^j\equiv1/3$. Then the full frontier is
\[
\delta^\ast_{\bar m}(\omega)=
\begin{cases}
  +\infty, & 0\le\bar m<t,\\[3pt]
  \max\{0,L_1(\bar m),L_2(\bar m)\}, & \bar m\ge t,
\end{cases}
\]
where
\[
  L_1(\bar m):=
  \frac{(1+2\varepsilon)t-\varepsilon(4-\varepsilon)\bar m}
       {1-\varepsilon^2},
  \qquad
  L_2(\bar m):=
  \frac{t-3\varepsilon\bar m}{1-\varepsilon}.
\]
Every branch is attained. Appendix~\ref{app:capped} gives the three-agent attaining matrix for the first positive branch, the capped equal-split construction for the second, and the corresponding exact construction for the zero branch.
\end{enumerate}
\end{theorem}

\begin{proof}
See Appendix~\ref{app:capped}.
\end{proof}

Part (i) follows directly from the identity: only the $n$ diagonal caps enter, and $\max_j|T_j|$ bounds the $k$-weighted imbalance. For $n\ge4$, part (ii) makes Corollary~\ref{cor:capfrontier} exact over the entire interior exposure range:
\[
  \bigl\{(\bar m,\bar\delta)\;:\;\bar m\ge t,\ \bar\delta\ge0,\ \ (1-\varepsilon)\bar\delta+n\varepsilon\bar m\ge t\bigr\}.
\]
For $n=3$,
\[
L_1(\bar m)-L_2(\bar m)
=\frac{\varepsilon[t-(1-4\varepsilon)\bar m]}{1-\varepsilon^2}.
\]
If $\varepsilon<1/7$, $L_1$ binds until the lines meet at $\bar m=t/(1-4\varepsilon)$, after which $L_2$ reaches zero at $t/(3\varepsilon)$. At $\varepsilon=1/7$ the kink coincides with that endpoint; if $\varepsilon>1/7$, $L_1$ remains binding until it reaches zero at the exact cap in \eqref{eq:full_exact_frontier}. The additional line is an allocation constraint: with three agents, there can be too few off-diagonal positions to complete the adjustment within the cap. The kink therefore marks a change in the binding linear constraint, not in preferences or beliefs. For $n\ge4$, symmetry spreads the residual imbalance uniformly across rules. Figure~\ref{fig:capped_frontiers} illustrates the two geometries.

\begin{figure}[H]
\centering
\begingroup
\pgfplotsset{
  capped axis/.style={
    width=\linewidth,
    height=.72\linewidth,
    xmin=.75, xmax=4.05,
    ymin=-.04, ymax=1.06,
    axis lines=left,
    axis line style={black!75},
    tick align=outside,
    tick style={black!75},
    ytick={0,.25,.5,.75,1},
    yticklabels={$0$,$1/4$,$1/2$,$3/4$,$1$},
    ymajorgrids,
    grid style={black!10},
    label style={font=\small},
    tick label style={font=\scriptsize},
    legend style={draw=none,fill=white,fill opacity=.90,text opacity=1,
      font=\scriptsize,at={(.98,.97)},anchor=north east,
      legend cell align=left,row sep=-1pt},
    every axis plot/.append style={line cap=round,line join=round}
  },
  capped exact/.style={very thick,blue!60!black},
  capped scaled/.style={thick,red!65!black,dashed},
  capped bound/.style={thick,black!60,densely dotted}
}

\begin{subfigure}[t]{.485\textwidth}
\centering
\begin{tikzpicture}
\begin{axis}[
  capped axis,
  xlabel={normalized cap $\bar m/t$},
  ylabel={worst-rule absolute imbalance $\delta_{\bar m}^{\ast}/t$},
  xtick={1,2.5,3.25,4},
  xticklabels={$1$,$5/2$,$13/4$,$4$}
]
  \addplot[draw=none,fill=black!8,forget plot]
    coordinates {(.75,-.04) (1,-.04) (1,1.06) (.75,1.06)} \closedcycle;
  \node[rotate=90,font=\scriptsize,text=black!60]
    at (axis cs:.865,.51) {infeasible};
  \addplot[capped exact]
    coordinates {(1,.6666667) (2.5,0) (4.05,0)};
  \addlegendentry{exact frontier}
  \addplot[capped scaled]
    coordinates {(1,1) (3.25,0)};
  \addlegendentry{scaled family}
\end{axis}
\end{tikzpicture}
\caption{Four agents: one binding segment.}
\end{subfigure}
\hfill
\begin{subfigure}[t]{.485\textwidth}
\centering
\begin{tikzpicture}
\begin{axis}[
  capped axis,
  xlabel={normalized cap $\bar m/t$},
  xtick={1,1.6666667,3.3333333,4},
  xticklabels={$1$,$5/3$,$10/3$,$4$}
]
  \addplot[draw=none,fill=black!8,forget plot]
    coordinates {(.75,-.04) (1,-.04) (1,1.06) (.75,1.06)} \closedcycle;
  \node[rotate=90,font=\scriptsize,text=black!60]
    at (axis cs:.865,.51) {infeasible};
  \addplot[capped exact]
    coordinates {(1,.8181818) (1.6666667,.5555556)
      (3.3333333,0) (4.05,0)};
  \addlegendentry{exact frontier}
  \addplot[capped scaled]
    coordinates {(1,1) (4,0)};
  \addlegendentry{scaled family}
  \addplot[capped bound]
    coordinates {(1,.7777778) (1.6666667,.5555556)};
  \addlegendentry{identity bound}
  \draw[black!40,thin,densely dashed]
    (axis cs:1.6666667,-.04) -- (axis cs:1.6666667,.5555556);
  \addplot[only marks,mark=*,mark size=1.6pt,blue!60!black,forget plot]
    coordinates {(1.6666667,.5555556)};
  \node[anchor=south west,font=\scriptsize,text=black!70]
    at (axis cs:1.70,.575) {kink};
\end{axis}
\end{tikzpicture}
\caption{Three agents: an additional tight-cap segment.}
\end{subfigure}
\endgroup
\caption{Capped repair frontiers at $\varepsilon=0.1$, normalized by the single-winner surplus $t$. Solid curves give exact values, dashed curves the scaled family, and the shaded region $\bar m/t<1$ is infeasible. With three agents, the additional tight-cap segment creates the kink at $\bar m/t=5/3$; the dotted segment is the surplus-absorption bound before it binds.}
\label{fig:capped_frontiers}
\end{figure}
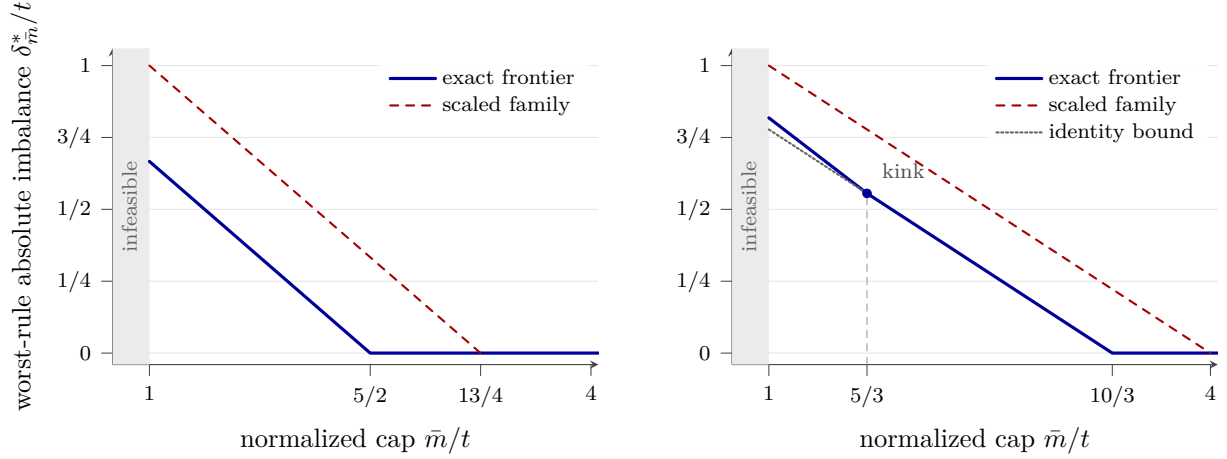

Convexity propagates the surplus-absorption branch from single-winner states to every nonnegative transfer profile on the domains stated in Corollary~\ref{cor:capped_anyx}. It does not propagate the tight-cap three-agent formula as an exact surplus-only value.

\begin{corollary}[Surplus-only capped frontier on its stated domain]\label{cor:capped_anyx}
Let the preserving weights be generated by \eqref{eq:Deltaqi}--\eqref{eq:qstar}, let $k^j\equiv1/n$, and let $\varepsilon_i^j=\varepsilon\in(0,1)$ for all $j\neq i$. Fix a state $\omega$ whose transfer vector is componentwise nonnegative with surplus $s(\omega)>0$. Then
\[
  \delta^\ast_{\bar m}(\omega)\;=\;\frac{\bigl[s(\omega)-n\varepsilon\bar m\bigr]_+}{1-\varepsilon}
\]
for every $\varepsilon\in(0,1)$ and $\bar m\ge s(\omega)$ when $n\ge4$. When $n=3$, the same conclusion holds for $\varepsilon\in(0,1/7]$ and $\bar m\ge s(\omega)/(1-4\varepsilon)$.
\end{corollary}

\begin{proof}
Write $x:=x(\omega)$ and $s:=s(\omega)$. The lower bound is part (i) of Theorem~\ref{thm:capped_frontier} with $s$ in place of $t$: Lemma~\ref{lem:trace} uses only preservation and the diagonal caps, so the derivation applies verbatim at any state. For the upper bound, argue as in Corollary~\ref{cor:anyx}. For each agent $i$, permutation symmetry applies the relevant single-winner construction to $s\,e_i$. On the positive branch this is the capped equal-split family, and every column imbalance equals $\bigl(s-n\varepsilon\bar m\bigr)/(1-\varepsilon)$. For $n\ge4$ and $\varepsilon\ge1/n$, the thick-exposure family instead gives exact balance already at cap $s$. Writing $x=\sum_i(x_i/s)(s\,e_i)$ and mixing the corresponding families entrywise preserves dominance and preservation, keeps every entry within $[-\bar m,\bar m]$, and gives the claimed common imbalance.
\end{proof}

For three agents below the cap domain of Corollary~\ref{cor:capped_anyx}, the capped value need not depend only on aggregate surplus; Remark~\ref{rem:three_not_surplus_only} gives an explicit counterexample. More generally, for caps between $\max_ix_i(\omega)$ and $s(\omega)$ the mixing argument is unavailable; the lower bound of Theorem~\ref{thm:capped_frontier}(i), with $s(\omega)$ in place of $t$, continues to apply.

}

Corollary~\ref{cor:capfrontier} and the constant family yield the two bounds used in the following vanishing-exposure result.

\begin{corollary}[Vanishing exposure]\label{cor:vanishing_exposure}
Fix a state $\omega$ with $s:=s(\omega)>0$, a full-support reference $k$, and a sequence of common-contamination exposures $\varepsilon_\tau\in(0,1)$ such that $\varepsilon_\tau\to0$. Let the preserving weights at each $\tau$ be generated by \eqref{eq:common_contamination_weights}. Let $m_\tau^\ast(\omega)$ and $\delta_{\bar m,\tau}^\ast(\omega)$ denote the exact-repair cap and the capped repair value under exposure $\varepsilon_\tau$. Even if the designer chooses a different preserving family at every $\tau$,
\[
  m_\tau^\ast(\omega)
  \;\ge\;\frac{s}{n\varepsilon_\tau}
  \;\longrightarrow\;+\infty.
\]
Moreover, for every fixed $\bar m\ge\max_i|x_i(\omega)|$,
\[
  \frac{\bigl[s-n\varepsilon_\tau\bar m\bigr]_+}{1-\varepsilon_\tau}
  \;\le\;\delta_{\bar m,\tau}^\ast(\omega)
  \;\le\;s,
  \qquad\text{and hence}\qquad
  \delta_{\bar m,\tau}^\ast(\omega)\longrightarrow s.
\]
Thus exact preservation, exact balance, and uniformly bounded realized transfers cannot coexist along a vanishing-exposure path.
\end{corollary}

\begin{proof}
Under exact balance, Lemma~\ref{lem:trace} gives $\sum_i x_i^i=s/\varepsilon_\tau$, so the absolute transfer cap is at least $s/(n\varepsilon_\tau)$. Corollary~\ref{cor:capfrontier} gives the lower bound on the capped value. For the upper bound, the constant family $x_i^j=x_i$ satisfies dominance and preservation, fits the stated cap, and leaves $T_j=s$ under every label.
\end{proof}

The exposure path is exogenous, so this result is a reduced-form comparative static. A structural learning model would determine that path endogenously; along any vanishing-exposure path, period-by-period recalibration cannot rescue bounded exact repair.

\long\def\HeterogeneousExposureResults{%
\subsection{Heterogeneous exposures: who pays for exactness}\label{sec:hetero}

A common exposure is a symmetric benchmark, not a requirement. This subsection instead allows agent-specific exposures $\varepsilon_i$, each constant across agent~$i$'s non-adverse labels, with weights
\[
  w_i:=\frac{\varepsilon_i}{1-\varepsilon_i}\;\ge\;0,
\]
so feasibility is governed by whether aggregate effective exposure $\sum_i\varepsilon_i(1-k^i)$ is positive. With a full-support reference, this is equivalent to some agent having positive exposure; the identity then dictates who carries the burden.

\begin{lemma}[Surplus-absorption identity under heterogeneous exposures]\label{lem:devidentity}
Fix a state, suppress it, and let $\varepsilon_i^j=\varepsilon_i\in[0,1)$ for all $j\neq i$. Any family satisfying \eqref{eq:preserve} obeys
\[
  s=\sum_{j\in\N}k^jT_j+\sum_{i\in\N}w_i\bigl(x_i^i-x_i\bigr);
  \qquad\text{under exact balance,}\quad
  \sum_{i\in\N}w_i\bigl(x_i^i-x_i\bigr)=s .
\]
\end{lemma}

\begin{proof}
Identical to the proof of Lemma~\ref{lem:trace}, carrying $\varepsilon_i$: preservation for agent $i$ rearranges to $\sum_jk^jx_i^j=(x_i-\varepsilon_ix_i^i)/(1-\varepsilon_i)$, and summing over $i$,
\[
  \sum_{j\in\N}k^jT_j+\sum_{i\in\N}w_i(x_i^i-x_i)
  =\sum_{i\in\N}\frac{(x_i-\varepsilon_ix_i^i)+\varepsilon_i(x_i^i-x_i)}{1-\varepsilon_i}
  =s. \qedhere
\]
\end{proof}

The identity also gives the necessary heterogeneous cap--imbalance bound
\[
  \max_j|T_j|
  \ge s-\sum_{i\in\N}w_i(\bar m-x_i)
\]
whenever $x_i^i\le\bar m$ for every $i$. No exact heterogeneous capped frontier is claimed. When $w_i\equiv\varepsilon/(1-\varepsilon)$, the bound reduces to Theorem~\ref{thm:capped_frontier}(i).

Aggregate exposure alone does not determine incidence: dispersed exposure spreads the adjustment across rows, whereas concentrated exposure loads it onto fewer agents. Theorem~\ref{thm:oneagent} gives the extensive margin---whether any effective off-diagonal exposure exists---and the following results identify the intensive burden at positive exposure.

\begin{corollary}[Burden concentration]\label{cor:pinned}
In the setting of Lemma~\ref{lem:devidentity} with $\sum_iw_i>0$, any family that preserves and balances exactly at a state with surplus $s\ge0$ satisfies
\[
  \max_{i:\,w_i>0}\ \bigl(x_i^i-x_i\bigr)\;\ge\;\frac{s}{\sum_{i\in\N}w_i}\,,
\]
and the burden falls entirely on exposed agents: unexposed agents contribute nothing to the identity and serve only as intermediaries. In particular, if exactly one agent $i_0$ is exposed, then the identity pins down her adverse-label payment at every state:
\[
  x_{i_0}^{i_0}(\omega)\;=\;x_{i_0}(\omega)+\frac{1-\varepsilon_{i_0}}{\varepsilon_{i_0}}\,s(\omega).
\]
It exceeds her original payment by $(1-\varepsilon_{i_0})/\varepsilon_{i_0}$ times the surplus.
\end{corollary}

\begin{proof}
Immediate from $\sum_iw_i(x_i^i-x_i)=s$: the maximum of the terms with $w_i>0$ is at least their weighted average $s/\sum_iw_i$, and with a single positive weight the single term equals $s$.
\end{proof}

The reference vector is the disclosed baseline used to construct the lower bounds. It is not assumed to be the objective selection frequency, a Bayesian belief, or a social-welfare weight. The next result is therefore an incidence and accounting identity under that baseline. It becomes an objective expected-transfer comparison only if the reference vector has an independent frequency interpretation.

\begin{proposition}[Reference-frequency accounting identity]\label{prop:reference_compensation}
Fix a state and suppress it. Let the preserving weights be generated from a common reference $k\in\Delta(\N)$ and off-diagonal exposures satisfying $\varepsilon_i^j=\varepsilon_i\in[0,1)$ for every $j\neq i$, and write $w_i:=\varepsilon_i/(1-\varepsilon_i)$. Suppose a payment-rule family satisfies dominance \eqref{eq:monotone}, preservation \eqref{eq:preserve}, and exact balance \eqref{eq:bb}. Define agent $i$'s reference-frequency accounting rebate by
\[
  \mathcal{R}_i^k:=x_i-\sum_{j\in\N}k^j x_i^j.
\]
Then
\[
  \mathcal{R}_i^k=w_i\bigl(x_i^i-x_i\bigr)\ge0
  \qquad\text{for every }i,
  \qquad
  \sum_{i\in\N}\mathcal{R}_i^k=s.
\]
Thus every agent's reference-weighted payment is weakly lower than under the pure rule, and the aggregate accounting difference equals the original surplus. Worst-case evaluated utilities remain unchanged. An exposed agent's accounting rebate equals her contribution to surplus absorption; if only one agent is exposed, her accounting rebate equals the entire surplus.
\end{proposition}

\begin{proof}
Write $\bar x_i^k:=\sum_jk^jx_i^j$. The preserving weight in \eqref{eq:qstar} and \eqref{eq:preserve} give
\[
  x_i=(1-\varepsilon_i)\bar x_i^k+\varepsilon_i x_i^i.
\]
Dominance implies $\bar x_i^k\le x_i^i$, hence $x_i\le x_i^i$, and rearranging gives $x_i-\bar x_i^k=w_i(x_i^i-x_i)\ge0$. Exact balance gives
\[
  \sum_i\mathcal{R}_i^k=\sum_ix_i-\sum_jk^j\sum_ix_i^j=s.
\]
The displayed payment comparison follows directly. Lemma~\ref{lem:qstar} and preservation leave the worst-case evaluated transfer equal to $x_i$. If $k$ is independently interpreted as the objective label frequency, label independence of the allocation turns the accounting difference into an expected-utility difference of the same amount; that additional interpretation is not imposed here.
\end{proof}

\begin{remark}[Expected incidence under an objective label law]\label{rem:endpoint_objective_law}
Suppose an analyst has an objective label law $\pi\in\Delta(\N)$, even though that law is not certified to the agents at the reporting date. For the canonical endpoint family,
\[
  x_i-\sum_{j\in\N}\pi^j\widehat x_i^j
  =\frac{1-\pi^i}{n-1}s\ge0,
  \qquad
  \sum_{i\in\N}\left(x_i-\sum_{j\in\N}\pi^j\widehat x_i^j\right)=s.
\]
Thus an agent with $\pi^i<1$ has a strictly positive expected realized-utility gain at a positive-surplus state, while preservation leaves her max--min value unchanged. This separates subjective incentive preservation from expected realized incidence under an independently specified data-generating law. It is a participant-side comparison; whether eliminating operator revenue is a social gain depends on the institution's retention constraint and the value assigned to that revenue.
\end{remark}

The individual formula is specific to the canonical family. Proposition~\ref{prop:endpoint_accounting_general} shows that nonnegativity and the aggregate identity hold for every dominant, preserving, exactly balanced endpoint family. Complete enforceable label-contingent insurance can eliminate label exposure only by replacing the compound mechanism with a balanced pure payment rule; Remark~\ref{rem:complete_insurance} gives the side-contract equivalence and its incentive caveat.

Together, these results show the feasibility discontinuity at zero effective off-diagonal exposure and, at positive effective exposure, identify both the burden and its reference-frequency accounting incidence.

}

\paragraph{Heterogeneous exposure.}
Theorem~\ref{thm:oneagent} gives the extensive feasibility margin in the main text. Appendix~\ref{sec:hetero} derives the complementary intensive margin: a heterogeneous surplus-absorption identity identifies which exposed agents must carry the adverse-label burden and gives the corresponding reference-frequency accounting incidence.

% ========================================================================
\section{Participation}\label{sec:transfers}
% ========================================================================

Pointwise repair preserves max--min participation, not participation under every realized rule. This section proves that the stronger requirement fails under full-support preserving weights and recovers at the support-only boundary; Appendix~\ref{app:transferbound} derives quantitative bounds.

\subsection{Evaluated and rule-by-rule participation}\label{subsec:participation_concepts}

Fix a state $\omega$ and write
\[
  y_i(\omega):=v_i(g(\omega),\omega),
  \qquad
  V(\omega):=\sum_{i\in\N}y_i(\omega).
\]
A payment-rule family satisfies \emph{rule-by-rule ex post IR} at $\omega$ if
\[
  y_i(\omega)-x_i^j(\omega)\ge 0
  \quad\text{for all }i,j.
\]
Rule-by-rule participation is stronger than max--min participation: it treats every label as a settlement the agent must be willing and able to honor and therefore captures hard liability limits. A family can preserve the evaluated participation value yet violate a realized settlement constraint. The results below are algebraic and do not use dominance, although dominance supports the worst-case interpretation.

\long\def\QuantitativeParticipationBound{%
\subsection{A sharp participation bound}\label{subsec:participation_bound}

\begin{proposition}[Sharp participation bound]\label{prop:rule_ir_sharp}
Fix a state $\omega$ and let the preserving weights be generated by \eqref{eq:Deltaqi}--\eqref{eq:qstar} with a full-support common reference $k$ and off-diagonal exposures satisfying $\varepsilon_i^j=\varepsilon_i\in[0,1)$ for $j\neq i$. If a payment-rule family preserves the original transfers at $\omega$, is exactly budget balanced at $\omega$, and satisfies rule-by-rule ex post IR at $\omega$, then
\[
  s(\omega)\;\le\;\sum_{i\in\N}w_i\bigl(y_i(\omega)-x_i(\omega)\bigr),
  \qquad w_i=\frac{\varepsilon_i}{1-\varepsilon_i}\,.
\]
With common exposure $\varepsilon_i\equiv\varepsilon$ this reads
\[
  s(\omega)\;\le\;\varepsilon\,V(\omega).
\]
The constant $\varepsilon$ cannot be improved: if $n\ge2$, $k^j\equiv1/n$, $y_i(\omega)\equiv y>0$, and $x_i(\omega)=\varepsilon y$ for every $i$ (so that $s=\varepsilon V$), the family $x_i^i=y$ and $x_i^j=-y/(n-1)$ for $j\neq i$ preserves, balances exactly, and satisfies rule-by-rule ex post IR with equality on the diagonal.
\end{proposition}

\begin{proof}
Only the diagonal IR constraints are needed. Under the stipulated weights, preservation rearranges to
\[
  \sum_{j\in\N}k^j x_i^j
  =x_i-w_i\bigl(x_i^i-x_i\bigr).
\]
Summing over agents and using exact balance gives
\[
  s=\sum_iw_i\bigl(x_i^i-x_i\bigr).
\]
Rule-by-rule IR gives $x_i^i(\omega)\le y_i(\omega)$, and therefore
\[
  s\le\sum_iw_i\bigl(y_i-x_i\bigr).
\]
With $\varepsilon_i\equiv\varepsilon$ the right side is $\frac{\varepsilon}{1-\varepsilon}(V-s)$, and $s\le\frac{\varepsilon}{1-\varepsilon}(V-s)$ rearranges to $s\le\varepsilon V$. For tightness, the displayed family balances exactly ($y-(n-1)\cdot\frac{y}{n-1}=0$ in every rule) and satisfies \eqref{eq:monotone} and rule-by-rule IR ($-\frac{y}{n-1}<y$ and $x_i^i=y_i$); it preserves because, with $k^j\equiv1/n$,
\[
  q_i^{i\ast}\,y-(1-\varepsilon)\,\tfrac{n-1}{n}\cdot\tfrac{y}{n-1}
  =\Bigl[\tfrac{1+(n-1)\varepsilon}{n}-\tfrac{1-\varepsilon}{n}\Bigr]y
  =\varepsilon y
  =x_i(\omega). \qedhere
\]
\end{proof}

Proposition~\ref{prop:rule_ir_sharp} is a necessary diagnostic, not a sufficient condition for arbitrary profiles. Under common exposure, rule-by-rule participation limits recycled surplus to the ambiguity share of aggregate realized value. More generally, only exposed agents' valuation slack is weighted: an unexposed agent contributes nothing, while an exposed agent with little slack may bind. By Proposition~\ref{prop:reference_compensation}, each exposed agent's contribution is also her reference-frequency accounting rebate. Positive aggregate slack remains insufficient unless the remaining preservation, balance, and participation constraints can also be met.

}

The main text now turns to the stronger qualitative failure in standard Vickrey states.

\subsection{Full-support impossibility}\label{subsec:participation_impossibility}

\begin{theorem}[Full-support impossibility for rule-by-rule participation]\label{thm:full_support_participation_impossibility}
Fix a state $\omega$ and suppose $q_i^{j\ast}>0$ for every $i,j$.
Suppose there is an agent $h$ such that
\[
  y_h(\omega)>0,\qquad
  y_i(\omega)=0 \text{ and } x_i(\omega)=0 \text{ for all } i\neq h,
  \qquad
  x_h(\omega)>0.
\]
Then no payment-rule family can simultaneously preserve the original transfers at $\omega$, be exactly budget balanced at $\omega$, and satisfy rule-by-rule ex post IR at $\omega$.
\end{theorem}

\begin{proof}
For every $i\neq h$, rule-by-rule IR implies $x_i^j(\omega)\le 0$ for all $j$.
Preservation gives
\[
  \sum_{j\in\N}q_i^{j\ast}x_i^j(\omega)=x_i(\omega)=0.
\]
Since $q_i^{j\ast}>0$ for every $j$ and every term in the weighted average is weakly negative, it follows that $x_i^j(\omega)=0$ for all $j$ and all $i\neq h$.
Exact budget balance then implies $x_h^j(\omega)=0$ for every rule $j$.
But preservation for agent $h$ would require
\[
  \sum_{j\in\N}q_h^{j\ast}x_h^j(\omega)=x_h(\omega)>0. \qedhere
\]
\end{proof}

\subsection{The support-only boundary}\label{subsec:participation_endpoint}

\begin{corollary}[Sharpness at the support-only endpoint]\label{cor:endpoint_ir_sharpness}
Fix a full-support reference $k$ and the common-contamination family
\[
  \Delta q_i(\varepsilon)
  =
  \bigl\{(1-\varepsilon)k+\varepsilon q:q\in\Delta(\N)\bigr\},
  \qquad \varepsilon\in[0,1].
\]
Fix a state satisfying the payment and value-profile hypotheses of Theorem~\ref{thm:full_support_participation_impossibility}, and suppose additionally that $x_h(\omega)\le y_h(\omega)$, so the original payment rule is ex post individually rational at that state. A payment-rule family satisfying dominance, preservation, exact budget balance, and rule-by-rule ex post individual rationality at the state exists if and only if $\varepsilon=1$.
\end{corollary}

\begin{proof}
For every $\varepsilon<1$, the adverse-label worst-case weight is
\[
  q_i^\ast(\varepsilon)=(1-\varepsilon)k+\varepsilon e_i.
\]
Because $k$ has full support, every component of every $q_i^\ast(\varepsilon)$ is positive. Theorem~\ref{thm:full_support_participation_impossibility} rules out the required family. At $\varepsilon=1$, $q_i^\ast=e_i$. Apply \eqref{eq:frequency_unrestricted_repair} at the fixed state: the resulting payments satisfy dominance, preservation, and exact budget balance there. Every realized payment is weakly below its original counterpart, so the assumed ex post individual rationality of the pure rule implies rule-by-rule ex post individual rationality.
\end{proof}

In a single-item Vickrey auction with positive revenue, the winner has positive realized value and payment, while losers have zero realized value and VCG payments. Theorem~\ref{thm:full_support_participation_impossibility} therefore makes the failure in Example~\ref{ex:vcg_vickrey3} construction-independent and stronger than a transfer-magnitude bound. Rule-by-rule participation forces losers' payments to be nonpositive; full-support preservation makes them all zero, and exact balance then eliminates the winner's positive preserved payment.

Full support is natural under a full-support reference at every interior exposure: every settlement rule matters for preservation. If some labels receive zero preserving weight, their transfers may escape the averaging argument. Corollary~\ref{cor:endpoint_ir_sharpness} establishes sharpness along the common-contamination path, not logical minimality for every pattern of zero weights. The discontinuity concerns whether non-adverse labels are behaviorally priced.

% ========================================================================
\section{Conclusion}\label{sec:discussion}
% ========================================================================
Pointwise preservation carries the original mechanism's dominant-strategy incentives into an exactly balanced payment-rule menu. Support-only ambiguity permits exact repair with small fixed menus and can preserve both the original transfer cap and participation under every realized rule. Positive frequency floors leave feasibility intact but raise transfer-cap requirements; with full-support preserving distributions, rule-by-rule participation fails in standard Vickrey states. These results separate feasibility, liquidity, and realized participation.

Implementation requires a precommitted, report-blind menu and an operational input that participants cannot manipulate. Exact balance removes the operator's direct monetary stake in the label, but not participants' preferences over labels or other nonmonetary motives. The results remain conditional on pointwise preservation and the maintained worst-case evaluation and, for the feasibility characterization, on the fixed-adverse-label architecture. The sharp frontiers apply to their stated benchmark domains. Repetition, learning, or complete label-contingent contracting may therefore require the mechanism to be redesigned.

% ========================================================================
% APPENDICES
% ========================================================================
\clearpage
\appendix
\numberwithin{equation}{section}

% ========================================================================

\section{Participation and transfer-feasibility bounds}\label{app:transferbound}
% ========================================================================

\QuantitativeParticipationBound

\subsection{A general participation bound}

Retain the statewise value notation of Section~\ref{sec:transfers}. The following bound allows arbitrary exposure patterns and uses only the row interval jointly imposed by rule-by-rule IR and exact balance. Proposition~\ref{prop:rule_ir_sharp} gives the sharper incidence-sensitive bound when exposures may vary across agents but are constant across each agent's non-adverse labels.

\begin{proposition}[Rule-by-rule ex post IR implies a thickness bound]\label{prop:rule_ir_bound}
Fix a state $\omega$ and suppose the ambiguity profile is generated by \eqref{eq:Deltaqi}--\eqref{eq:qstar}. Suppose $y_i(\omega)\ge0$ for every agent. If a payment-rule family preserves the original transfers, balances exactly, and satisfies rule-by-rule ex post IR at $\omega$, then
\[
  s(\omega)\le R\,V(\omega),
\]
where $R=\sum_i q_i^{i\ast}-1=\sum_j k^j\sum_{i\neq j}\varepsilon_i^j$.
\end{proposition}

\begin{proof}
Rule-by-rule IR gives $x_i^j(\omega)\le y_i(\omega)$. Exact balance also gives
\[
  x_i^j(\omega)=-\sum_{h\neq i}x_h^j(\omega)
  \ge-\sum_{h\neq i}y_h(\omega)
  =y_i(\omega)-V(\omega),
\]
so every entry in row $i$ lies in the interval $[y_i(\omega)-V(\omega),y_i(\omega)]$. Using \eqref{eq:qstar}, preservation implies
\[
  s(\omega)=\sum_jk^j\sum_ix_i^j(\omega)
  +\sum_jk^j\sum_{i\neq j}\varepsilon_i^j
  \bigl(x_i^i(\omega)-x_i^j(\omega)\bigr).
\]
The first term is zero, and each within-row spread has absolute value at most $V(\omega)$, which yields the claim.
\end{proof}

\subsection{Two-agent transfer-cap benchmark}

The next benchmark isolates the inverse-exposure lower bound and the unique preserving family at a single-winner state.

\begin{proposition}[Transfer caps and ambiguity thickness in the symmetric two-agent case]\label{prop:blowup_lb}
Consider a fixed state $\omega$, two agents, and two rule labels. Suppose $k^1=k^2=\frac12$, $\varepsilon_i^i=0$, and $\varepsilon_i^j=\varepsilon\in(0,1)$ for $i\neq j$, so that
\[
q_1^\ast=\Bigl(\frac{1+\varepsilon}{2},\frac{1-\varepsilon}{2}\Bigr),
\qquad
q_2^\ast=\Bigl(\frac{1-\varepsilon}{2},\frac{1+\varepsilon}{2}\Bigr).
\]
Let $x(\omega)=(x_1,x_2)$ have surplus $s:=x_1+x_2\ge0$. Any family satisfying exact rule-by-rule balance and preservation obeys
\[
\max_{i\in\{1,2\},\,j\in\{1,2\}} |x_i^j(\omega)|
\ge\frac{s}{2\varepsilon}.
\]
Consequently, a cap $\bar m>0$ requires $\varepsilon\ge s/(2\bar m)$. If $x(\omega)=(t,0)$ with $t>0$, the feasible family is unique and equals
\[
x^1(\omega)=\Bigl(\frac{1+\varepsilon}{2\varepsilon}t,\,-\frac{1+\varepsilon}{2\varepsilon}t\Bigr),
\qquad
x^2(\omega)=\Bigl(-\frac{1-\varepsilon}{2\varepsilon}t,\,\frac{1-\varepsilon}{2\varepsilon}t\Bigr),
\]
so its maximum absolute transfer is $(1+\varepsilon)t/(2\varepsilon)$.
\end{proposition}

\begin{proof}
Exact balance gives $x^1=(a,-a)$ and $x^2=(b,-b)$. Preservation gives
\[
\frac{1+\varepsilon}{2}a+\frac{1-\varepsilon}{2}b=x_1,
\qquad
-\frac{1-\varepsilon}{2}a-\frac{1+\varepsilon}{2}b=x_2.
\]
Adding yields $\varepsilon(a-b)=s$, hence $2\max\{|a|,|b|\}\ge s/\varepsilon$. The cap condition follows. When $x=(t,0)$, solving the two equations gives the displayed unique family.
\end{proof}

\subsection{General transfer-feasibility bound}

The symmetric two-agent benchmark in Proposition~\ref{prop:blowup_lb} isolates the $1/\varepsilon$ blow-up in the clearest possible way. The next proposition records a more general necessary condition that allows arbitrary $n$, heterogeneous ambiguity, and a uniformly bounded ex post imbalance.

\begin{proposition}[A general lower bound with transfer caps and approximate balance]\label{prop:general_transfer_lb}
Fix a state $\omega$ and suppose the ambiguity profile satisfies \eqref{eq:Deltaqi}--\eqref{eq:qstar}. Let
\[
  s(\omega):=\sum_{i\in\N} x_i(\omega),
  \qquad
  T_j(\omega):=\sum_{i\in\N} x_i^j(\omega),
  \qquad
  R:=\sum_{i\in\N} q_i^{i\ast}-1
  =\sum_{j\in\N} k^j \sum_{i\neq j} \varepsilon_i^j.
\]
Consider any payment-rule family $\{x^j(\omega)\}_{j\in\N}$ satisfying
\[
\sum_{j\in\N} q_i^{j\ast}x_i^j(\omega)=x_i(\omega)
\quad\text{for all } i\in\N.
\]
If for some $\bar m,\bar\delta\ge 0$,
\[
|x_i^j(\omega)|\le \bar m
\quad\text{for all } i,j,
\qquad
|T_j(\omega)|\le \bar\delta
\quad\text{for all } j,
\]
then
\[
  s(\omega)\le \bar\delta + 2R\,\bar m.
\]
Equivalently,
\[
  \bar\delta \ge [\,s(\omega)-2R\,\bar m\,]_+.
\]
In particular, if $\bar\delta=0$ and $R>0$, then
\[
  \bar m \ge \frac{s(\omega)}{2R}.
\]
\end{proposition}

\begin{proof}
Using \eqref{eq:qstar}, rewrite the preservation identities as
\begin{align*}
  x_i(\omega)
  &= q_i^{i\ast}x_i^i(\omega) + \sum_{j\neq i}(1-\varepsilon_i^j)k^j x_i^j(\omega) \\
  &= \sum_{j\in\N} k^j x_i^j(\omega)
   + \sum_{j\neq i} \varepsilon_i^j k^j \bigl(x_i^i(\omega)-x_i^j(\omega)\bigr)
\end{align*}
for each agent $i$. Summing over $i$ yields
\[
  s(\omega)
  = \sum_{j\in\N} k^j T_j(\omega)
  + \sum_{j\in\N} k^j \sum_{i\neq j} \varepsilon_i^j \bigl(x_i^i(\omega)-x_i^j(\omega)\bigr).
\]
Therefore,
\begin{align*}
  s(\omega)
  &\le \sum_{j\in\N} k^j |T_j(\omega)|
   + \sum_{j\in\N} k^j \sum_{i\neq j} \varepsilon_i^j \bigl|x_i^i(\omega)-x_i^j(\omega)\bigr| \\
  &\le \bar\delta + 2\bar m \sum_{j\in\N} k^j \sum_{i\neq j} \varepsilon_i^j
   = \bar\delta + 2R\bar m.
\end{align*}
The equivalent lower bound on $\bar\delta$ is immediate. If $\bar\delta=0$ and $R>0$, rearranging gives $\bar m \ge s(\omega)/(2R)$.
\end{proof}

The bound is vacuous when $s(\omega)\le0$. At positive-surplus states it specializes to Proposition~\ref{prop:blowup_lb} in the symmetric two-agent exact-balance benchmark: $R=\varepsilon$ and $\bar\delta=0$ give $\bar m\ge s/(2\varepsilon)$. The benchmark result additionally identifies the unique preserving family when $x(\omega)=(t,0)$.

Proposition~\ref{prop:general_transfer_lb} uses neither dominance nor common exposure and bounds each adverse-label spread by two capped transfers, which explains its factor of two. When exposures may vary across agents but are constant across each agent's non-adverse labels, the master decomposition instead closes into the exact identity of Lemma~\ref{lem:devidentity}. With common exposure this yields $s(\omega)\le(1-\varepsilon)\bar\delta+n\varepsilon\bar m$. With a uniform reference, the associated lower bound truncated at zero is exact throughout the interior exposure range over the feasible cap domain in the single-winner benchmark for $n\ge4$, and Corollary~\ref{cor:capped_anyx} extends the same value to componentwise nonnegative payment vectors on its stated cap domain. With three agents, Lemma~\ref{lem:three_allocation} supplies an additional allocation line. Outside these stated domains, the inequality remains only a necessary bound. Thus the general inequality is a quick screen, while the identities and constructions in Section~\ref{sec:price} price repair sharply.

\section{Feasibility extensions and classification}\label{app:feasibility_supp}

\subsection{Why necessity is class-relative}\label{app:scope_counterexample}

The dominance condition in Theorem~\ref{thm:bbtrans} is substantive. Let $n=2$, $q_1^\ast=(0.7,0.3)$, $q_2^\ast=(0.3,0.7)$, and $x=(-1,0)$. The unique family satisfying balance and preservation alone is
\[
  x^1=\bigl(-\tfrac74,\tfrac74\bigr),
  \qquad
  x^2=\bigl(\tfrac34,-\tfrac34\bigr).
\]
It repairs the deficit but violates dominance because $x_1^2>x_1^1$. Dropping the fixed-adverse-label restriction therefore enlarges the set of repairable rules. Under the lower-bound foundation, dominance is also what supports a report-independent worst-case minimizer.

\subsection{Calibration and unknown exposure}\label{app:calibration}

\begin{proposition}[Calibration sensitivity of the closed form]\label{prop:calibration_sensitivity}
Suppose \eqref{eq:closed_form} is constructed using preserving weights $q^\ast$ satisfying the reference-wedge condition with $R>0$. Let $\tilde q_i\in\Delta(\N)$ be any report-independent evaluation weight. Then
\[
  \sum_{j\in\N}\tilde q_i^j x_i^j(\omega)-x_i(\omega)
  =\frac{\tilde q_i^i-q_i^{i\ast}}{R}\,s(\omega).
\]
Consequently, if $|\tilde q_i^i-q_i^{i\ast}|\le\eta_i$ and $|s(\omega)|\le\bar s$ over the relevant report profiles, the evaluated-transfer error is at most $\eta_i\bar s/R$. Every incentive margin changes by at most $2\eta_i\bar s/R$, and every participation payoff by at most $\eta_i\bar s/R$.
\end{proposition}

\begin{proof}
Substitution in \eqref{eq:closed_form} gives
\[
\begin{aligned}
  \sum_j\tilde q_i^j x_i^j(\omega)
  &=x_i(\omega)+\frac{s(\omega)}{R}
    \sum_j\tilde q_i^j\bigl(\mathbf{1}\{j=i\}-q_i^{i\ast}\bigr)\\
  &=x_i(\omega)+\frac{\tilde q_i^i-q_i^{i\ast}}{R}s(\omega).
\end{aligned}
\]
The payoff bounds follow because quasi-linear utility subtracts transfers. A report-payoff difference contains two transfer errors, while participation contains one.
\end{proof}

The support-only family's sensitivity is especially transparent. Under common-contamination weights
\[
  q_i^\ast(\varepsilon)=(1-\varepsilon)k+\varepsilon e_i,
\]
evaluating the endpoint family \eqref{eq:frequency_unrestricted_repair} gives
\begin{equation}\label{eq:endpoint_calibration_error}
  \sum_{j\in\N}q_i^{j\ast}(\varepsilon)\widehat x_i^j(\omega)-x_i(\omega)
  =
  -\frac{(1-\varepsilon)(1-k^i)}{n-1}\,s(\omega).
\end{equation}
With a full-support reference and positive surplus, the error vanishes at the support-only endpoint but not at any narrower common exposure. Because the surplus varies with reports, this is generally an incentive-relevant error rather than a harmless constant. Thus the endpoint is a distinct support regime, not a worst-case plug-in value for an unknown interior exposure.

For the closed form, off-diagonal misspecification enters only through the simplex constraint and the resulting mass on the adverse label. Proposition~\ref{prop:no-uniform-calibration} gives the exact two-exposure impossibility in the main text. Proposition~\ref{prop:calibration_sensitivity} and Proposition~\ref{prop:eta_robust} quantify the corresponding approximate errors and feasibility boundary.

\subsection{Bayesian incentive and interim participation}\label{app:bayesian_extension}

\BayesianExtensionResult

\HomogeneousAndClassificationResults

\subsection{Details for the reference-set classification}\label{app:wedge_details}

\begin{proof}[Proof of Lemma~\ref{lem:wedgechar}]
By definition, $K^+(q^\ast)=\{k\in\Delta(\N):M^j\le k^j\le q_j^{j\ast}\ \forall j\}$. Membership implies $1=\sum_jk^j\ge\sum_jM^j$. Conversely, suppose $\sum_jM^j\le1$. Since row $i$ sums to one,
\[
  \sum_jM^j\ge\sum_{j\neq i}q_i^{j\ast}+M^i=(1-q_i^{i\ast})+M^i,
\]
so $M^i\le q_i^{i\ast}$ for every $i$. Moreover, $q_j^{j\ast}\ge M^j\ge q_1^{j\ast}$ for $j\neq1$, hence $\sum_jq_j^{j\ast}\ge1$. The box $\prod_j[M^j,q_j^{j\ast}]$ therefore meets the hyperplane whose coordinates sum to one, and every intersection point lies in $K^+(q^\ast)$. The argument for $K^-(q^\ast)=\{k\in\Delta(\N):q_j^{j\ast}\le k^j\le m^j\ \forall j\}$ is symmetric. Necessity is $1\le\sum_jm^j$. For sufficiency, the same row-sum argument gives $m^i\ge q_i^{i\ast}$ for every $i$. In addition, $q_j^{j\ast}\le m^j\le q_1^{j\ast}$ for $j\neq1$, while $q_1^{1\ast}$ is the first coordinate of row~1; hence $\sum_jq_j^{j\ast}\le\sum_jq_1^{j\ast}=1$. The box $\prod_j[q_j^{j\ast},m^j]$ therefore meets the unit-sum hyperplane.
\end{proof}

\begin{corollary}[Case (iii) collapses; case (iv) requires crossed weights]\label{cor:crossed}
If $K^+(q^\ast)$ and $K^-(q^\ast)$ are both nonempty, then $M^j=m^j$ for every $j$. Consequently, both sets can be empty only if two opponents disagree about the weight of some label; in particular, this is impossible when $n=2$.
\end{corollary}

\begin{proof}
Both sets nonempty gives $\sum_jM^j\le1\le\sum_jm^j$. Since $M^j\ge m^j$ componentwise, every inequality binds. With two agents, each label has a single opponent and therefore $M^j=m^j$ automatically; Lemma~\ref{lem:wedgechar} then makes at least one reference set nonempty.
\end{proof}

\subsection{Crossed preserving weights: the betting branch}\label{sec:crossed}

Case (ii) of Theorem~\ref{thm:general_feasibility} is the sign-reversed branch for statewise deficits. Case (iv) is the only branch with no aggregate-sign restriction and hence the only one that can repair both surplus and deficit states. Corollary~\ref{cor:crossed} shows that it requires at least three agents and genuinely crossed preserving weights.

\begin{example}[A cyclic crossed profile repairs a pure deficit]\label{ex:cyclic}
Let $n=3$ and
\[
q_1^\ast=(0.2,0.7,0.1),\qquad
q_2^\ast=(0.1,0.2,0.7),\qquad
q_3^\ast=(0.7,0.1,0.2).
\]
Then $M^j=0.7$ and $m^j=0.1$, so both reference sets are empty and case (iv) applies. For the uniform deficit rule $x=(-d,-d,-d)$ with $d>0$, define, with indices modulo three,
\[
x_i^{\,i}=x_i^{\,i-1}=\tfrac{10}{11}d,
\qquad
x_i^{\,i+1}=-\tfrac{20}{11}d.
\]
Every column sums to zero. For each agent, preservation follows from
\[
0.2\cdot\tfrac{10}{11}d
+0.7\cdot\bigl(-\tfrac{20}{11}d\bigr)
+0.1\cdot\tfrac{10}{11}d=-d,
\]
and dominance holds because $x_i^{\,i+1}<x_i^{\,i-1}=x_i^{\,i}$.
\end{example}

Every realized rule balances, yet every agent's evaluation registers a receipt. The construction is therefore a money pump built from disagreement, in the spirit of belief-based extraction \citep{CremerMcLean1985,CremerMcLean1988}.

\begin{remark}[Where crossed profiles can come from]\label{rem:crossed_foundation}
The common-reference foundation cannot generate case (iv), because its reference lies in $K^+(q^\ast)$. Crossed profiles require agent-specific primitive references. Let $r_1=(0.2,0.7,0.1)$, $r_2=(0.1,0.2,0.7)$, and $r_3=(0.7,0.1,0.2)$. Replacing $k$ by $r_i$ in \eqref{eq:Deltaqi}, with common exposure $\varepsilon$, gives $\tilde q_i^{j\ast}=(1-\varepsilon)r_i^j$ for $j\neq i$ and $\tilde q_i^{i\ast}=\varepsilon+(1-\varepsilon)r_i^i$. Thus $M^j=0.7(1-\varepsilon)$ and $m^j=0.1(1-\varepsilon)$, so case (iv) persists for $\varepsilon\in[0,11/21)$. This branch relies on non-common references over the designer's own device and should delimit, rather than extend, the common-reference economic interpretation.
\end{remark}

% ========================================================================

\section{Additional price analysis}\label{app:price_supp}

\FullCappedFrontier

\HeterogeneousExposureResults

\subsection{Spread and approximate preservation}\label{app:spread_robustness}

\SecondarySpreadAndRobustness

\subsection{Endpoint accounting and complete insurance}

\begin{proposition}[General endpoint accounting identity]\label{prop:endpoint_accounting_general}
Fix a state and suppose the preserving weight is $q_i^\ast=e_i$ for every agent. For any BB transformation $\{z^j\}_{j\in\N}$ and any accounting vector $k\in\Delta(\N)$, define
\[
  \mathcal{R}_i^k:=x_i-\sum_{j\in\N}k^jz_i^j.
\]
Then
\[
  \mathcal{R}_i^k\ge0
  \quad\text{for every }i,
  \qquad
  \sum_{i\in\N}\mathcal{R}_i^k=s.
\]
For the canonical family \eqref{eq:frequency_unrestricted_repair},
\[
  \mathcal{R}_i^k=\frac{1-k^i}{n-1}s,
\]
and hence $\mathcal{R}_i^k=s/n$ for every agent under the uniform accounting vector.
\end{proposition}

\begin{proof}
Preservation at $e_i$ gives $z_i^i=x_i$, and dominance gives $z_i^j\le x_i$ for every $j$. Hence $\mathcal{R}_i^k\ge0$. Exact balance yields
\[
  \sum_i\mathcal{R}_i^k
  =
  \sum_i x_i-\sum_jk^j\sum_i z_i^j
  =s.
\]
For the canonical family,
\[
  \sum_jk^j\widehat x_i^j
  =
  k^ix_i+(1-k^i)
  \left(x_i-\frac{s}{n-1}\right),
\]
which gives the individual formula. This direct argument is needed because the odds ratio in Proposition~\ref{prop:reference_compensation} is undefined at the endpoint.
\end{proof}

\begin{remark}[Complete label-contingent insurance]\label{rem:complete_insurance}
Suppose the grand coalition can commit before the label is realized to an enforceable side contract contingent on the report profile and the label, without side-transfer caps or transaction costs. For any statewise balanced pure payment rule $c(r)$, the receipt
\[
  z_i^j(r):=x_i^j(r)-c_i(r)
\]
is zero sum and makes the net payment equal to the label-independent rule $c_i(r)$. Conversely, any zero-sum side contract that makes net payments label independent produces a statewise balanced pure rule.

The canonical reference-zero hedge sets $c_i(r):=\sum_jk^jx_i^j(r)$. Proposition~\ref{prop:reference_compensation} gives $c_i(r)\le x_i(r)$ and $\sum_i[x_i(r)-c_i(r)]=s(r)$. This comparison is attractive at a fixed report profile, but it does not implement the original pure rule: if anticipated when reporting, the effective rule is $c$, whose incentive properties are not inherited from $x$. Requiring $c=x$ is compatible with balance only at zero-surplus states. This is an insurance-feasibility observation under unrestricted multilateral contracting, not an equilibrium prediction.

\end{remark}

\subsection{Frontier comparisons}

\begin{remark}[The frontier under approximate preservation]\label{rem:eta_frontier}
Let $m^\ast_\eta(\omega)$ replace exact preservation in $m^\ast(\omega)$ by the tolerance in Proposition~\ref{prop:eta_robust}. On the equal-split domain $\varepsilon\le\hat\varepsilon(n)$, at a single-winner state and for $\eta\in[0,t/n)$,
\[
  \frac{t-n\eta}{n\varepsilon}
  \le m^\ast_\eta(\omega)
  \le\frac{t-\eta}{n\varepsilon}.
\]
The lower bound follows by applying Lemma~\ref{lem:trace} to the exactly preserved vector after the errors are absorbed; the upper bound applies the equal-split construction to $(t-\eta)e_1$. Both converge to $t/(n\varepsilon)$ as $\eta\to0$.
\end{remark}

\begin{remark}[Relation to the general bound and the two-agent benchmark]\label{rem:frontier_comparison}
Here $R=(n-1)\varepsilon$. On the equal-split branch, the exact cap $t/(n\varepsilon)$ is $2(n-1)/n$ times the general floor $s/(2R)$ from Proposition~\ref{prop:general_transfer_lb}; the difference is the factor lost when a two-sided spread is bounded by two capped transfers. In the two-agent, two-label single-winner benchmark, Proposition~\ref{prop:blowup_lb} gives the sharper rigid cap $(1+\varepsilon)t/(2\varepsilon)$. Under common contamination, enriching the finite label set cannot improve the equal-split branch: Proposition~\ref{prop:arbitrary_menu_identity} gives the same lower bound, which the $n$-label construction attains. Elsewhere the exact frontier also reflects the winner's preservation floor and, with three agents, the allocation bound in \eqref{eq:full_exact_frontier}.
\end{remark}

\begin{remark}[Two minimax criteria]\label{rem:two_criteria}
The closed form \eqref{eq:closed_form} minimizes the maximum adverse-label spread, whereas Theorem~\ref{thm:frontier} minimizes the largest absolute transfer. The criteria conflict. In Example~\ref{ex:vcg_vickrey3} at $\varepsilon=1/10$ and $s=6$, the closed form has minimal spread $30$ but cap $24$, while the equal-split family has optimal cap $20$ but spread $36\tfrac23$. Limited liability therefore favors the latter, while dispersion minimization favors the former. Remark~\ref{rem:capped_position} extends the comparison to the full cap--imbalance frontier.
\end{remark}

\begin{remark}[The three-agent capped value is not always surplus-only]\label{rem:three_not_surplus_only}
Let $n=3$, fix $t>0$ and $\varepsilon\in(0,1/7]$, and set $\bar m=s=t$. The single-winner capped value is $t(1-\varepsilon)/(1+\varepsilon)$. For the symmetric profile $x=(t/3,t/3,t/3)$, the matrix with diagonal $t$ and every off-diagonal entry $-\varepsilon t/(1-\varepsilon)$ satisfies preservation, dominance, and the cap and attains $t(1-3\varepsilon)/(1-\varepsilon)$. The values differ by $4\varepsilon^2t/(1-\varepsilon^2)>0$, so below the domain of Corollary~\ref{cor:capped_anyx} the distribution of surplus can matter.
\end{remark}

\subsection{Linear-programming form and duality}

\begin{remark}[Dual formulation]\label{rem:lp}
For fixed $\omega$, $m^\ast(\omega)$ is the value of a finite linear program. Whenever the primal is feasible, standard duality gives
\[
  m^\ast(\omega)=\max\Bigl\{\textstyle\sum_i\lambda_i x_i(\omega):
  \varphi\ge0,\ \textstyle\sum_{i,j}|\Psi_{ij}|\le1\Bigr\},
\]
where $\Psi_{ij}=\lambda_iq_i^{j\ast}+\mu_j-\varphi_{ij}$ for $j\neq i$ and $\Psi_{ii}=\lambda_iq_i^{i\ast}+\mu_i+\sum_{\ell\neq i}\varphi_{i\ell}$. Here $\lambda$, $\mu$, and $\varphi$ are the multipliers on preservation, balance, and dominance. Stationarity in the capped transfers identifies $\Psi_{ij}$ with the difference of the multipliers on the two cap constraints, while stationarity in the cap makes their total mass one.

For any wedge witness $k$, the flat dual point $\lambda_i=1/(2R)$, $\mu_j=-k^j/(2R)$, and $\varphi=0$ is feasible because
\[
  \sum_{i,j}|\Psi_{ij}|=\frac{1}{2R}
  \left[\sum_i(q_i^{i\ast}-k^i)
  +\sum_i\sum_{j\neq i}(k^j-q_i^{j\ast})\right]=1.
\]
Its value is $s(\omega)/(2R)$, the bound in Proposition~\ref{prop:general_transfer_lb}. Moreover, $m^\ast(\omega)=+\infty$ exactly when the primal is infeasible, which is the Farkas system in Appendix~\ref{app:general_feasibility}. Theorem~\ref{thm:general_feasibility} therefore describes the finiteness pattern of the frontier.
\end{remark}

\subsection{The scaled closed form}

\begin{proposition}[Scaled closed form]\label{prop:scaled}
Suppose $R=\sum_iq_i^{i\ast}-1>0$, and fix a state $\omega$ with $s(\omega)\ge0$. For every $\gamma\in[0,1]$, the family \eqref{eq:scaled} satisfies \eqref{eq:monotone} and \eqref{eq:preserve}, and
\[
  T_j(\omega;\gamma)=(1-\gamma)s(\omega)
  \qquad\text{for every }j.
\]
Consequently, for every $\bar m\ge\max_i|x_i(\omega)|$,
\[
  \delta^\ast_{\bar m}(\omega)
  \le\bigl(1-\gamma_{\bar m}(\omega)\bigr)s(\omega),
  \qquad
  \gamma_{\bar m}(\omega):=\max\Bigl\{\gamma\in[0,1]:
  \max_{i,j}|x_i^j(\omega;\gamma)|\le\bar m\Bigr\}.
\]
\end{proposition}

\begin{proof}
The adjustment in \eqref{eq:scaled} has zero $q_i^\ast$-mean, and
$x_i^i(\omega;\gamma)-x_i^j(\omega;\gamma)=\gamma s(\omega)/R\ge0$. Summing over agents gives
\[
  T_j(\omega;\gamma)=s(\omega)+\gamma\frac{1-\sum_iq_i^{i\ast}}{R}s(\omega)
  =(1-\gamma)s(\omega).
\]
The maximum absolute entry is a continuous maximum of affine functions of $\gamma$. At zero it is no larger than $\bar m$, so the feasible set of scaling factors is a nonempty closed interval and $\gamma_{\bar m}$ is attained. Its family is feasible in \eqref{eq:deltastar}, proving the bound.
\end{proof}

\begin{remark}[Where the scaled closed form sits]\label{rem:capped_position}
At a single-winner state $x=(t,0,\ldots,0)$ with $t>0$, uniform reference, and common exposure, the scaled family's largest entry is $t\bigl(1+\gamma(1-\varepsilon)/(n\varepsilon)\bigr)$. Thus, over its relevant cap range, it delivers
\[
  \max_j|T_j|
  =\frac{t\bigl(1+(n-1)\varepsilon\bigr)-n\varepsilon\bar m}{1-\varepsilon}.
\]
On the equal-split domain $\varepsilon\le\hat\varepsilon(n)$, the comparison is especially simple. For $n\ge4$, the scaled value exceeds the exact frontier by $(n-1)\varepsilon t/(1-\varepsilon)$ throughout the binding interval. For $n=3$ and $t>0$, the exact optimality gap is
\[
\begin{cases}
\dfrac{\varepsilon\bigl[(1+2\varepsilon)t+(1-4\varepsilon)\bar m\bigr]}{1-\varepsilon^2},
  & t\le\bar m\le\dfrac{t}{1-4\varepsilon},\\[9pt]
\dfrac{2\varepsilon t}{1-\varepsilon},
  & \dfrac{t}{1-4\varepsilon}<\bar m\le\dfrac{t}{3\varepsilon},\\[9pt]
\dfrac{(1+2\varepsilon)t-3\varepsilon\bar m}{1-\varepsilon},
  & \dfrac{t}{3\varepsilon}<\bar m
    \le\dfrac{(1+2\varepsilon)t}{3\varepsilon}.
\end{cases}
\]
The expressions agree at both kinks, and the gap is zero above the last displayed cap. Outside the equal-split domain, Theorem~\ref{thm:capped_frontier} still gives the exact comparator, but the displayed three-piece gap formula no longer applies. The scaled family reaches exact balance only above the minimum exact-balance cap. Hence it is an implementable benchmark, not a cap-optimal repair.
\end{remark}

\begin{remark}[One value surface]\label{rem:capped_lp}
For fixed $\omega$, the feasible cap--imbalance pairs form a closed convex set. The minimal cap is its zero-imbalance section, while the capped repair value gives its fixed-cap sections. Theorem~\ref{thm:capped_frontier} and Corollary~\ref{cor:frequency_unrestricted} trace these sections at uniform-reference single-winner states. More generally, the flat dual point in Remark~\ref{rem:lp} gives $\delta^\ast_{\bar m}(\omega)\ge[\,s(\omega)-2R\bar m\,]_+$. Fixed-cap sections become finite once the cap contains the original transfers and reach zero exactly in the feasible branches of Theorem~\ref{thm:general_feasibility}.
\end{remark}

\section{Omitted proofs}\label{app:proofs}

\subsection{Proof of Proposition~\ref{prop:eta_robust}}\label{app:eta_robust}

\begin{proof}
\emph{Necessity.}
Fix $\omega$ and write
\[
  \delta_i
  :=
  x_i(\omega)-\sum_jq_i^{j\ast}x_i^j(\omega)
  \in[-\eta,\eta].
\]
The family exactly preserves the vector $\bigl(x_i(\omega)-\delta_i\bigr)_i$. Applying the identity from the proof of Proposition~\ref{prop:closedform_optimal} to that vector gives
\[
  s(\omega)-\sum_{i\in\N}\delta_i
  =
  \sum_{i\in\N}\sum_{j\neq i}
  \bigl(k^j-q_i^{j\ast}\bigr)
  \bigl(x_i^i(\omega)-x_i^j(\omega)\bigr)
  \ge0.
\]
The reference wedge makes the coefficients nonnegative and dominance makes the spreads nonnegative. Hence $s(\omega)\ge\sum_i\delta_i\ge-n\eta$.

\emph{Sufficiency.}
Define $\tilde x_i(\omega):=x_i(\omega)+\max\{0,-s(\omega)\}/n$. The rule $\tilde x$ is statewise non-deficit, so Theorem~\ref{thm:bbtrans} provides a BB transformation of it. That family satisfies dominance and exact balance, and its preservation error relative to $x$ is exactly $\max\{0,-s(\omega)\}/n\le\eta$ at every state.
\end{proof}

\subsection{Proof of Corollary~\ref{cor:alpha_repair}}\label{app:alpha_repair}

\begin{proof}
For $e_i$, contamination gives
\[
  a_i=\rho_i^\alpha(e_i)=(1-\varepsilon)k^i+\alpha\varepsilon,
  \qquad
  A_\rho-1=\varepsilon(n\alpha-1).
\]
Thus the pointwise construction and argument of Proposition~\ref{prop:rho_repair} prove sufficiency and give the displayed family.

For necessity and the spread bound, consider any preserving, dominant, exactly balanced family. Put
\[
  d_i:=x_i^i,\qquad
  m_i:=\min_jx_i^j,\qquad
  h_i:=d_i-m_i,\qquad
  y_i^j:=x_i^j-m_i.
\]
Dominance implies $h_i\ge0$, $y_i^j\ge0$, and $y_i^i=h_i$. Exact column balance gives, for every label $j$,
\[
  0=\sum_i x_i^j=\sum_i m_i+\sum_i y_i^j.
\]
Averaging these identities across labels yields
\[
  \sum_i m_i=-\frac1n\sum_{i,j}y_i^j.
\]
Weighting the column-balance identities by $k$ instead yields
\[
  \sum_i\sum_jk^jx_i^j
  =
  \sum_jk^j\sum_i x_i^j
  =0.
\]
Under common contamination,
\[
  \rho_i^\alpha\bigl((x_i^j)_j\bigr)
  =
  (1-\varepsilon)\sum_jk^jx_i^j
  +\varepsilon\bigl[\alpha d_i+(1-\alpha)m_i\bigr].
\]
Preservation and the weighted balance identity therefore give
\[
  \frac{s}{\varepsilon}
  =
  \alpha\sum_i h_i+\sum_i m_i
  =
  \alpha\sum_i h_i-\frac1n\sum_{i,j}y_i^j.
\]
Since $\sum_{i,j}y_i^j\ge\sum_i h_i$,
\[
  \frac{s}{\varepsilon}
  \le
  \left(\alpha-\frac1n\right)\sum_i h_i.
\]
If $\alpha\le1/n$, the right-hand side is nonpositive, contradicting $s>0$. Thus $\alpha>1/n$. Only then, using $\sum_i h_i\le n\max_i h_i$, do we obtain
\[
  \frac{s}{\varepsilon}
  \le
  (n\alpha-1)\max_i h_i.
\]
Rearranging gives the stated spread bound. The construction has $h_i=s/[\varepsilon(n\alpha-1)]$ for every agent and therefore attains equality.
\end{proof}

\subsection{Proof of Theorem~\ref{thm:general_feasibility}}\label{app:general_feasibility}

\begin{proof}
Fix a state $\omega$ and suppress $\omega$ in notation.
We seek numbers $\{x_i^j\}_{i,j\in\N}$ satisfying:
(i) $x_i^j-x_i^i\le0$ for all $i$ and $j\neq i$;
(ii) $\sum_i x_i^j=0$ for all $j$;
(iii) $\sum_j q_i^{j\ast}x_i^j=x_i$ for all $i$.
By a standard linear alternative (the system $\{Az=b,\ Cz\le0\}$ is feasible if and only if every pair $(y,\phi)$ with $\phi\ge0$ and $y^{\top}A+\phi^{\top}C=0$ satisfies $y^{\top}b\ge0$), this system is feasible if and only if every collection of multipliers
$\alpha\in\mathbb{R}^n$, $\beta\in\mathbb{R}^n$, and
$\phi\in\mathbb{R}_+^{n(n-1)}$ satisfying
\begin{align*}
  \alpha_i q_i^{j\ast}+\beta_j+\phi_{i,j}&=0
  &&\text{for all } i,\; j\neq i,\\
  \alpha_i q_i^{i\ast}+\beta_i-\sum_{\ell\neq i}\phi_{i,\ell}&=0
  &&\text{for all } i
\end{align*}
also satisfies $\sum_i\alpha_i x_i\ge0$.

Because $\phi_{i,j}\ge0$ enters the first set of equations additively, dual feasibility of $(\alpha,\beta)$ requires
\[
  \alpha_i q_i^{j\ast}+\beta_j\le0
  \quad\text{for all } i \text{ and } j\neq i,
\]
and, given any such $(\alpha,\beta)$, the multipliers $\phi_{i,j}=-(\alpha_i q_i^{j\ast}+\beta_j)\ge0$ are admissible provided the second set of equations holds. Substituting
$\phi_{i,j}=-(\alpha_i q_i^{j\ast}+\beta_j)$
into the second set of equations gives
\[
  \alpha_i+\sum_{\ell\in\N}\beta_\ell=0
  \quad\text{for all }i.
\]
Hence every dual-feasible vector has $\alpha_i\equiv\alpha$ for some scalar $\alpha$.

Suppose $\alpha>0$ and define $k^j:=-\beta_j/\alpha$. Then $\sum_j k^j=1$, and dividing the inequality display above by $\alpha$ gives
\[
  k^j\;\ge\;q_i^{j\ast}\quad\text{for all }i\neq j,
  \qquad\text{equivalently}\qquad
  k^j\;\ge\;M^j=\max_{i\neq j}q_i^{j\ast}\quad\text{for all }j.
\]
Note that the upper inequality $k^j\le q_j^{j\ast}$ is \emph{not} implied by dual feasibility: the diagonal equations are exhausted by the substitution above, so a positive dual scalar exists if and only if some $k$ satisfies $\sum_jk^j=1$ and $k^j\ge M^j$ for every $j$. Since $M^j\ge0$, such a vector exists if and only if $\sum_jM^j\le1$, and by Lemma~\ref{lem:wedgechar} this holds if and only if $K^+(q^\ast)\neq\varnothing$. Thus positive dual scalars exist if and only if $K^+(q^\ast)$ is nonempty.
Conversely, any $k\in K^+(q^\ast)$ and any $\alpha>0$ generate a dual-feasible vector directly, by setting
$\beta_j=-\alpha k^j$ and
$\phi_{i,j}=\alpha(k^j-q_i^{j\ast})\ge0$.

Similarly, if $\alpha<0$ and $k^j:=-\beta_j/\alpha$, dividing the inequality display by $\alpha<0$ reverses it:
\[
  k^j\;\le\;m^j=\min_{i\neq j}q_i^{j\ast}
  \quad\text{for all }j,
\]
together with $\sum_jk^j=1$. The intermediate vector need only satisfy these two conditions and may fall outside the simplex; such a vector exists if and only if $\sum_jm^j\ge1$, and Lemma~\ref{lem:wedgechar} converts this scalar condition into the existence of a genuine simplex vector in $K^-(q^\ast)$. Thus negative dual scalars exist if and only if $K^-(q^\ast)$ is nonempty; conversely any $k\in K^-(q^\ast)$ and $\alpha<0$ generate a dual-feasible vector by $\beta_j=-\alpha k^j$ and $\phi_{i,j}=-\alpha(q_i^{j\ast}-k^j)\ge0$.
In the case $\alpha=0$, the first set of equations forces $\beta_j=-\phi_{i,j}\le0$ for every $j$, while the second forces $\beta_i=\sum_{\ell\neq i}\phi_{i,\ell}\ge0$; hence $\beta=0$ and $\phi=0$, and the zero multiplier imposes no restriction on $x$.

The dual feasibility condition therefore reduces to the following scalar requirement:
$\alpha s\ge0$ for every $\alpha>0$ when $K^+$ is nonempty, and for every $\alpha<0$ when $K^-$ is nonempty.
This is equivalent to the four cases stated in the theorem:
if only $K^+$ is nonempty, $s\ge0$; if only $K^-$ is nonempty, $s\le0$; if both are nonempty, $s=0$; and if neither is nonempty, no nonzero aggregate multiplier exists and there is no aggregate-sign restriction.
Because the linear system separates across states, the statewise condition is necessary and sufficient for a BB transformation of the whole payment rule.
\end{proof}

\subsection{Proof of Proposition~\ref{prop:preserveIC}}\label{app:preserveIC}

\begin{proof}
Fix agent $i$. Proposition~\ref{prop:pointwise_incentive} gives equality \eqref{eq:pointwise_preserve_payoff} at every report profile. The same $q_i^\ast$ also minimizes every interim expectation. Indeed, for every $q_i\in\Delta q_i$ and every realization of the other agents' types, the integrand evaluated at $q_i^\ast$ is no larger than the integrand evaluated at $q_i$. Taking the conditional expectation preserves this inequality, so the minimum of the expectation is attained at $q_i^\ast$.

\textbf{Part (a): BIC.}
For every type $\theta_i$ and candidate report $r_i$, the preceding observation and \eqref{eq:preserve} give
\[
\begin{aligned}
&\min_{q_i \in \Delta q_i}
  \E_p\!\left[v_i(g(r_i,\theta_{-i}),\theta)
  -\sum_{j\in\N}q_i^j x_i^j(r_i,\theta_{-i})\mid\theta_i\right]\\
&\qquad=
  \E_p\!\left[v_i(g(r_i,\theta_{-i}),\theta)
  -x_i(r_i,\theta_{-i})\mid\theta_i\right].
\end{aligned}
\]
Thus every candidate report has exactly its pure-mechanism interim payoff, and the BIC constraints coincide.

\textbf{Part (b): IIR.}
\[
  \min_{q_i\in\Delta q_i}
  \E_p\!\left[v_i(g(\theta),\theta)-\sum_{j\in\N}q_i^j x_i^j(\theta)\mid\theta_i\right]
  =\E_p\!\left[v_i(g(\theta),\theta)-x_i(\theta)\mid\theta_i\right],
\]
so the IIR condition is identical under both mechanisms.
\end{proof}

\subsection{Proof of Lemma~\ref{lem:three_allocation}}\label{app:three_bound}

\begin{proof}
Put $q:=1-\varepsilon$ and $\kappa:=1+2\varepsilon$. Average any capped preserving family with the family obtained by simultaneously interchanging agents and labels $2$ and $3$. Preservation and the cap are retained, while the maximum absolute column imbalance cannot increase. The averaged matrix has the form
\[
  \begin{pmatrix}
    a&b&b\\
    c&d&f\\
    c&f&d
  \end{pmatrix},
\]
and preservation is equivalent to
\[
  \kappa a+2qb=3t,
  \qquad
  q(c+f)+\kappa d=0.
\]
Its column sums are $U:=a+2c$ and $V:=b+d+f$, with the latter repeated twice. Let $D:=\max\{|U|,|V|\}$. From $U\le D$,
\[
  c\le\frac{D-a}{2}.
\]
Loser preservation also gives
\[
  V=b-\frac{q}{\kappa}c+\frac{3\varepsilon}{\kappa}f.
\]
Using $f\ge-\bar m$, $a\le\bar m$, and $D\ge V$ therefore yields
\begin{align*}
  \left(1+\frac{q}{2\kappa}\right)D
  &\ge b+\frac{q}{2\kappa}a-\frac{3\varepsilon}{\kappa}\bar m\\
  &=\frac{3t}{2q}
    -\frac{3\varepsilon(2+\varepsilon)}{2q\kappa}a
    -\frac{3\varepsilon}{\kappa}\bar m\\
  &\ge
  \frac{3\bigl[\kappa t-\varepsilon(4-\varepsilon)\bar m\bigr]}
       {2q\kappa}.
\end{align*}
Since $1+q/(2\kappa)=3(1+\varepsilon)/(2\kappa)$, this is \eqref{eq:three_allocation_bound}. The original family obeys the same bound because symmetrization cannot increase the objective.
\end{proof}

\subsection{Proof of Theorem~\ref{thm:frontier}}\label{app:frontier}

\begin{proof}
Fix the state and suppress it; after relabeling, $x=(t,0,\dots,0)$ with $t>0$.

\emph{Universal lower bounds.}
Preservation for agent $1$ writes $t=\sum_jq_1^{j\ast}x_1^j$ as a convex combination of row $1$, so $\max_j|x_1^j|\ge t$. By Lemma~\ref{lem:trace}, exact balance forces $\sum_ix_i^i=t/\varepsilon$, so $\max_ix_i^i\ge t/(n\varepsilon)$. Neither step uses dominance.

\emph{The equal-split branch.}
Let $k^j\equiv1/n$. On the branch $n\ge4$ with $\varepsilon\le1/n$, or $n=3$ with $\varepsilon\le1/7$, use
\begin{align*}
  x_1^1&=\frac{t}{n\varepsilon}, &
  x_1^j&=\frac{\bigl[(n^2-n+1)\varepsilon-1\bigr]\,t}
                {n(n-1)\varepsilon(1-\varepsilon)}
                &&(j\neq1),\\[2pt]
  x_i^1&=-\frac{t}{n(n-1)\varepsilon}
                &&(i\neq1), &
  x_i^i&=\frac{t}{n\varepsilon}
                &&(i\neq1),\\[2pt]
  x_i^j&=-\frac{\bigl[(n^2-2n+2)\varepsilon+n-2\bigr]\,t}
                {n(n-1)(n-2)\varepsilon(1-\varepsilon)}
                &&(i\neq1,\ j\notin\{1,i\}).
\end{align*}
Put $A=D=t/(n\varepsilon)$ and denote the other three entries by $B,C,E$ in the order displayed above. Direct substitution gives row preservation and
\[
  A+(n-1)C=0,
  \qquad
  B+D+(n-2)E=0.
\]
Thus every column balances. Moreover,
\[
  B-C=\frac{t}{1-\varepsilon},
  \qquad
  B+D=
  \frac{\bigl[(n^2-2n+2)\varepsilon+n-2\bigr]t}
       {n(n-1)\varepsilon(1-\varepsilon)}>0,
\]
which verifies loser preservation and dominance once the cap inequalities hold. Clearly $|C|\le A$. One checks that $|B|\le A$ for $\varepsilon\le1/n$, and that $|E|\le A$ if and only if
\[
  \varepsilon\le
  \bar\varepsilon(n):=
  \frac{(n-2)^2}{2n^2-5n+4}.
\]
Here $\bar\varepsilon(3)=1/7$, while
\[
  n(n-2)^2-(2n^2-5n+4)=(n-1)^2(n-4)\ge0
\]
shows $\bar\varepsilon(n)\ge1/n$ for $n\ge4$. Hence the equal-split family is feasible with cap $t/(n\varepsilon)$ on exactly the branch stated in the theorem.

\emph{Four or more agents, thick exposure.}
For $n\ge4$ and $\varepsilon\ge1/n$, use the thick-exposure family
\begin{align*}
  x_1^j&=t &&(j\in\N),\\
  x_i^1&=-\frac{t}{n-1} &&(i\neq1),\\
  x_i^i&=\frac{(1-\varepsilon)t}{(n-1)\varepsilon} &&(i\neq1),\\
  x_i^j&=-\frac{[1+(n-2)\varepsilon]t}
                   {(n-1)(n-2)\varepsilon}
  &&(i\neq1,\ j\notin\{1,i\}).
\end{align*}
The winner's row is constant at $t$, so it is preserved. Column $1$ balances because
\[
  t+(n-1)\left(-\frac{t}{n-1}\right)=0,
\]
and for every other column the loser entries satisfy
\[
  \frac{(1-\varepsilon)t}{(n-1)\varepsilon}
  +(n-2)\left(
  -\frac{[1+(n-2)\varepsilon]t}
        {(n-1)(n-2)\varepsilon}\right)
  =-t.
\]
For a loser, the row sum is $-nt/(n-1)$, so
\[
  \frac{1-\varepsilon}{n}\sum_jx_i^j+\varepsilon x_i^i=0,
\]
which is preservation. The off-diagonal entries are negative and the diagonal is nonnegative, so dominance holds. The diagonal is at most $t$ if and only if $\varepsilon\ge1/n$; the last off-diagonal has magnitude at most $t$ if and only if $\varepsilon\ge1/(n-2)^2$. For $n\ge4$, the latter follows from $\varepsilon\ge1/n$. Thus the family attains cap $t$.

\emph{Three agents.}
Under exact balance, Lemma~\ref{lem:three_allocation} with $\max_j|T_j|=0$ implies
\[
  m^\ast\ge
  \frac{(1+2\varepsilon)t}{\varepsilon(4-\varepsilon)}=:m_3.
\]
Together with the diagonal bound this gives the maximum in \eqref{eq:full_exact_frontier}. The diagonal bound is the larger one exactly when $\varepsilon\le1/7$, where the equal-split construction attains it. If $\varepsilon\ge1/7$, use
\begin{equation}\label{eq:three_exact_thick}
\begin{pmatrix}
 m_3 & \dfrac{(7\varepsilon-1)t}{2\varepsilon(4-\varepsilon)}
     & \dfrac{(7\varepsilon-1)t}{2\varepsilon(4-\varepsilon)}\\[8pt]
 -m_3/2 & \dfrac{3(1-\varepsilon)t}{2\varepsilon(4-\varepsilon)} & -m_3\\[8pt]
 -m_3/2 & -m_3 & \dfrac{3(1-\varepsilon)t}{2\varepsilon(4-\varepsilon)}
\end{pmatrix}.
\end{equation}
Its distinct entries have
\[
  B=\frac{(7\varepsilon-1)t}{2\varepsilon(4-\varepsilon)},
  \qquad
  C=-\frac{m_3}{2},
  \qquad
  D=\frac{3(1-\varepsilon)t}{2\varepsilon(4-\varepsilon)}.
\]
Direct substitution verifies preservation and zero column sums. Also $0\le B\le m_3$ for $\varepsilon\in[1/7,1)$, and $0\le D\le m_3$ is equivalent to $\varepsilon\ge1/7$. Thus every entry lies in $[-m_3,m_3]$ and each diagonal dominates its row. This attains the remaining branch.
\end{proof}

\subsection{Proof of Theorem~\ref{thm:capped_frontier}}\label{app:capped}

\begin{proof}
Fix the state and suppress it. Preservation for the winner writes $t$ as a convex combination of her row, so no family fits a cap below $t$. For the universal lower bound, Lemma~\ref{lem:trace} gives
\[
  \sum_jk^jT_j
  =t-\frac{\varepsilon}{1-\varepsilon}\sum_i(x_i^i-x_i)
  \ge\frac{t-n\varepsilon\bar m}{1-\varepsilon}.
\]
Since $\max_j|T_j|\ge\sum_jk^jT_j$ and the objective is nonnegative, part~(i) follows.

\emph{Four or more agents.}
Suppose first that $\varepsilon\le1/n$ and $t\le\bar m\le t/(n\varepsilon)$. Put $\kappa:=1+(n-1)\varepsilon$ and use the capped equal-split family
\begin{align*}
  x_1^1&=\bar m, &
  x_1^j&=\frac{n\,t-\kappa\bar m}{(n-1)(1-\varepsilon)}
          &&(j\neq1),\\[2pt]
  x_i^1&=\frac{t-\kappa\bar m}{(n-1)(1-\varepsilon)}
          &&(i\neq1), &
  x_i^i&=\bar m
          &&(i\neq1),\\[2pt]
  x_i^j&=-\frac{t+(n-2)\kappa\bar m}
                  {(n-1)(n-2)(1-\varepsilon)}
          &&(i\neq1,\ j\notin\{1,i\}).
\end{align*}
Denote the five entries by $A,B,C,D,E$ in the order displayed, so $A=D=\bar m$.
Preservation is equivalent to
\[
  \frac{1-\varepsilon}{n}\sum_jx_i^j+\varepsilon x_i^i=x_i.
\]
The winner's row sum is $n(t-\varepsilon\bar m)/(1-\varepsilon)$, and every loser row sum is $-n\varepsilon\bar m/(1-\varepsilon)$, so preservation holds. Each column sum equals
\[
  \bar\delta:=\frac{t-n\varepsilon\bar m}{1-\varepsilon}\ge0.
\]

It remains to check the cap. The inequalities $B\le\bar m$ and $C\le0$ follow from $\bar m\ge t$. Because $n\ge4$ and $\varepsilon\le1/n\le1/4$, the lower bounds $B,C\ge-\bar m$ follow from $(n-1)(1-2\varepsilon)\ge1$. Finally,
\[
  |E|\le\bar m
  \quad\Longleftrightarrow\quad
  t\le c(n,\varepsilon)\bar m,
  \qquad
  c(n,\varepsilon):=(n-2)[(n-2)-2(n-1)\varepsilon].
\]
For $n\ge4$ and $\varepsilon\le1/n$,
\[
  c(n,\varepsilon)\ge1
  \quad\Longleftrightarrow\quad
  \varepsilon\le\frac{n-3}{2(n-2)},
\]
and the last inequality follows from $(n-1)(n-4)\ge0$. Hence the family is capped and dominant and attains the universal lower bound. At $\bar m=t/(n\varepsilon)$ it becomes the exactly balanced equal-split family, which remains feasible under every larger cap. If instead $\varepsilon\ge1/n$, the thick-exposure family is exactly balanced already under cap $t$; the positive part in part~(ii) is then zero for every $\bar m\ge t$.

\emph{Three agents.}
Put $q:=1-\varepsilon$ and $\kappa:=1+2\varepsilon$. Whenever $L_1$ is the positive maximum, use
\begin{equation}\label{eq:three_tight_family}
\begin{pmatrix}
  \bar m & B & B\\
  C & D & -\bar m\\
  C & -\bar m & D
\end{pmatrix},
\qquad
B:=\frac{3t-\kappa\bar m}{2q},\quad
D:=\frac{(3-2\varepsilon)\bar m-t}{2(1+\varepsilon)},\quad
C:=\bar m-\frac{\kappa}{q}D.
\end{equation}
Lemma~\ref{lem:three_allocation}, part~(i), and nonnegativity of the objective give the lower bound $\max\{0,L_1(\bar m),L_2(\bar m)\}$. Winner and loser preservation for \eqref{eq:three_tight_family} follow from
\[
  \kappa\bar m+2qB=3t,
  \qquad
  q(C-\bar m)+\kappa D=0,
\]
and all three column sums equal $L_1(\bar m)$. For $\bar m\ge t$, one has $D>0>C$ and
\begin{align*}
  \bar m-B&=\frac{3(\bar m-t)}{2q},&
  B+\bar m&=\frac{3t+(1-4\varepsilon)\bar m}{2q},\\
  \bar m-D&=\frac{t-(1-4\varepsilon)\bar m}{2(1+\varepsilon)},&
  C+\bar m&=
  \frac{\kappa t+(1-4\varepsilon)\bar m}{2q(1+\varepsilon)}.
\end{align*}
If $\varepsilon\le1/7$, these inequalities verify the cap and dominance for
\[
  t\le\bar m\le\frac{t}{1-4\varepsilon}.
\]
If $\varepsilon\ge1/7$, they do so for
\[
  t\le\bar m\le
  m_3:=\frac{\kappa t}{\varepsilon(4-\varepsilon)}.
\]
Indeed, for $\varepsilon<1/4$ one has $m_3\le t/(1-4\varepsilon)$; for $\varepsilon\ge1/4$ the potentially negative numerators remain nonnegative because
\[
  m_3\le\frac{\kappa t}{4\varepsilon-1},
  \qquad
  m_3\le\frac{3t}{4\varepsilon-1},
\]
both valid for $\varepsilon<1$. Thus the tight family attains $L_1$ throughout its binding positive range.

The two affine bounds differ by
\[
  L_1(\bar m)-L_2(\bar m)
  =\frac{\varepsilon[t-(1-4\varepsilon)\bar m]}
         {1-\varepsilon^2}.
\]
For $\varepsilon\le1/7$, they meet at $t/(1-4\varepsilon)$. From that point to $t/(3\varepsilon)$, the capped equal-split family specialized to $n=3$ satisfies the cap and dominance and attains $L_2$; at the endpoint it is exactly balanced. For $\varepsilon\ge1/7$, $L_1$ remains the positive maximum until it reaches zero at $m_3$, where the tight family becomes \eqref{eq:three_exact_thick}. The appropriate exactly balanced endpoint remains feasible under every larger cap. This proves part~(iii).
\end{proof}

\ifdefined\UTMDWorkingPaper
\else
\ifdefined\ArxivPreprint
\section*{Funding}
This work was supported by the Japan Society for the Promotion of Science (JSPS) through KAKENHI (Grants-in-Aid for Scientific Research), Grant Number JP21H04979, and by the Japan Science and Technology Agency (JST) through Exploratory Research for Advanced Technology (ERATO), Grant Number JPMJER2301.

\section*{Verification code}
The ancillary files include the computational verification scripts used to cross-check the partition, arbitrary-menu, calibration, frontier, and participation results.

\section*{Acknowledgments}
I thank Hitoshi Matsushima and Atsushi Iwasaki for guidance and helpful discussions. An earlier version of this paper circulated as UTMD Working Paper No.~130, ``Exact Budget-Balance Transformations via Payment-Rule Ambiguity'' (April 2026); the present manuscript is substantially revised and expanded. Any remaining errors are my own.

\section*{Use of generative AI tools}
During the preparation of this work, the author used OpenAI's ChatGPT to assist with language editing, literature cross-checking, manuscript organization, consistency review, and the generation of code for numerical and algebraic cross-checks. The author also used Refine to obtain an AI-assisted review of logic and internal consistency. The author independently reviewed and verified all outputs and suggestions, edited the manuscript as needed, and takes full responsibility for the content of this work.
\else
\section*{Funding}
This work was supported by the Japan Society for the Promotion of Science (JSPS) through KAKENHI (Grants-in-Aid for Scientific Research), Grant Number JP21H04979, and by the Japan Science and Technology Agency (JST) through Exploratory Research for Advanced Technology (ERATO), Grant Number JPMJER2301.
The funders had no role in the design of the study, the analysis or interpretation of the results, the preparation of the manuscript, or the decision to submit it for publication.

\section*{Declaration of competing interest}
The author has no relevant financial or non-financial interests to disclose.

\section*{Data availability}
No data were used for the research described in this article.

\section*{Code availability}
The computational verification scripts used to cross-check the partition, arbitrary-menu, calibration, frontier, and participation results are supplied as supplementary material.

\section*{Acknowledgments}
I thank Hitoshi Matsushima and Atsushi Iwasaki for guidance and helpful discussions. An earlier version of this paper circulated as UTMD Working Paper No.~130, ``Exact Budget-Balance Transformations via Payment-Rule Ambiguity'' (April 2026); the present manuscript is substantially revised and expanded. Any remaining errors are my own.

\section*{Declaration of generative AI and AI-assisted technologies in the manuscript preparation process}
During the preparation of this work, the author used OpenAI's ChatGPT to assist with language editing, literature cross-checking, manuscript organization, consistency review, and the generation of code for numerical and algebraic cross-checks. The author also used Refine to obtain an AI-assisted review of logic and internal consistency. The author independently reviewed and verified all outputs and suggestions, edited the manuscript as needed, and takes full responsibility for the content of the published article.
\fi
\fi

% Bibliography
\begingroup
\ifdefined\UTMDWorkingPaper
\else
\small
\fi
\bibliographystyle{plainnat}
\bibliography{references}
\endgroup

\end{document}